\documentclass{non-vldb}
\usepackage{graphicx}
\usepackage{balance}  

\usepackage[mathletters]{ucs}
\usepackage[utf8x]{inputenc}
\usepackage{tikz-cd}
\usepackage[noend,ruled,noline]{algorithm2e}

\usepackage{epstopdf}
\usepackage{mathtools}
\usepackage{etoolbox}
\usepackage{proof}

\usepackage{subcaption}
\usepackage{listings,color} 
\usepackage{microtype}
\usepackage{multirow}
\usepackage[hidelinks]{hyperref}

\usepackage{collectbox}
\usepackage{makecell}
\usepackage{algpseudocode}
\usepackage{hhline}
\usepackage{microtype}
\usepackage{paralist}
\usepackage{pifont}

\usepackage{tikz}
 \usetikzlibrary{graphs,quotes}
 \usetikzlibrary{arrows}
 \usetikzlibrary{shapes}
 
\allowdisplaybreaks

\newcommand{\revision}[1]{{\color{black}{#1}}}

\newcommand\deltaAdd[1]{\Delta^{+}_{#1}}
\newcommand\deltaDel[1]{\Delta^{-}_{#1}}

\newcommand\ruleeq{\mathrel{\mathord{:}-}}

\newcommand{\xmark}{\ding{53}}%
\newcommand{\yes}{\checkmark}%
\newcommand{\no}{\xmark}%

\newtheorem{theorem}{Theorem}[section]
\newtheorem{lemma}[theorem]{Lemma}
\newtheorem{proposition}[theorem]{Proposition}

\newtheorem{definition}{Definition}[section]

\newtheorem{example}{Example}[section]

\newtheorem{claim}{Claim}

\AtEndEnvironment{example}{\hfill{$\qed$}}%
\AtEndEnvironment{remark}{\hfill{$\qed$}}%
\AtEndEnvironment{proof}{\hfill{ }}%

\definecolor{mygreen}{HTML}{1B5E20}

\lstnewenvironment{Datalog}[1][]
{\lstset{
        language=prolog,
        mathescape = true,
        basicstyle=\ttfamily,
        columns=flexible,
        keepspaces=true,
        breaklines=true,
        captionpos=b, 
        belowcaptionskip= 0em,
		numbers=left, 
        xleftmargin=2em, 
        framexleftmargin=1.5em,
        stepnumber=1,
        numberfirstline=false,
        sensitive=true,
        extendedchars=\true,
        showstringspaces=false,
        commentstyle=\color{mygreen},
        #1,
}}{}

\lstnewenvironment{SQL}[1][]
{\lstset{
		language=SQL,
		mathescape = true,
		basicstyle=\ttfamily\footnotesize,
		columns=flexible,
		keepspaces=true,
		breaklines=true,
		captionpos=b, 
		belowcaptionskip= 0em,
		numbers=left, 
		xleftmargin=2em, 
		framexleftmargin=1.5em,
		stepnumber=1,
		numberfirstline=false,
		sensitive=true,
		extendedchars=\true,
		showstringspaces=false,
		commentstyle=\color{mygreen},
		#1,
}}{}

\vldbTitle{Programmable View Update Strategies on Relations}
\vldbAuthors{Van-Dang Tran, Hiroyuki Kato, Zhenjiang Hu}
\vldbDOI{https://doi.org/10.14778/3377369.3377380}
\vldbVolume{13}
\vldbNumber{5}
\vldbYear{2020}

\begin{document}


\title{Programmable View Update Strategies on Relations}



%
%
%
%

\numberofauthors{1} 

\author{
%
%
%
%
%
\alignauthor
Van-Dang Tran$^{3,1}$,
Hiroyuki Kato$^{1,3}$,
Zhenjiang Hu$^{2,1}$\\
\smallskip
       \affaddr{$^1$National Institute of Informatics, Japan}\\
       \affaddr{$^2$Peking University, China}\\
       \affaddr{$^3$The Graduate University for Advanced Studies, SOKENDAI, Japan}\\
       \email{\{dangtv, kato\}@nii.ac.jp, huzj@pku.edu.cn}
}


\maketitle

\begin{abstract}
View update is an important mechanism that allows updates on a view by translating them into the corresponding updates on the base relations. The existing literature has shown the ambiguity of translating view updates.
To address this ambiguity, we propose a robust language-based approach for making view update strategies programmable and validatable.
Specifically, we introduce a novel approach to use Datalog to describe these update strategies. We propose a validation algorithm to check the well-behavedness of the written Datalog programs. We present a fragment of the Datalog language for which our validation is both sound and complete. This fragment not only has good properties in theory but is also useful for solving practical view updates.
Furthermore, we develop an algorithm for optimizing user-written programs to efficiently implement updatable views in relational database management systems.
We have implemented our proposed approach.
The experimental results show that our framework is feasible and efficient in practice.
\end{abstract}

\section{Introduction}
\label{sec:introduction}

View update \cite{Bancilhon1981, Dayal1978, Dayal1982, Fagin1983, Keller1985} is an important mechanism in relational databases.
This mechanism allows updates on a view by translating them into the corresponding updates on the base relations \cite{Dayal1982}.
Consider a view $V$ defined by a query $get$ over the database $S$, as shown in Figure~\ref{fig:viewupdate}.
An update translator $T$ maps each update $u$ on $V$ to an update $T(u)$ on $S$ such that it is {\em well-behaved} in the sense that after the view update is propagated to the source, we will obtain the same view from the updated source, i.e., $u (V) = get (T(u)(S))$.
Given a view definition $get$, the known \textit{view update problem} \cite{Dayal1982} is to derive such an update translator $T$.

However, there is an ambiguity issue here. Because the query $get$ is generally not injective, there may be many update translations on the source database that can be used to reflect view update \cite{Dayal1978, Dayal1982}. This ambiguity makes view update an open challenging problem that has a long history in database research \cite{Fagin1983, Dayal1978, Dayal1982, Bancilhon1981, Keller1986, Keller1985, Kotidis2006, Kimelfeld2012, Masunaga2017, Larson1991, Langerak1990}. The existing approaches either impose too many syntactic restrictions on the view definition $get$ that allow for limited unambiguous update propagation \cite{Dayal1982, Bohannon2006,Bancilhon1981, Keller1984, Lechtenborger2003, Langerak1990, Masunaga1984, Masunaga2017, Masunaga2018} or provide dialogue mechanisms for users to manually choose update translations with users' interaction \cite{Keller1986, Larson1991}. In practice, commercial database systems such as PostgreSQL \cite{postgresql} provide very limited support for updatable views such that even a simple union view cannot be updated.

In this paper, we propose a new approach for solving the view updating problem practically and correctly. The key idea is to provide a formal language for people to directly program their view update strategies. On the one hand, this language can be considered a formal treatment of Keller's dialogue \cite{Keller1986}, but on the other hand, it is unique in that it can fully determine the behavior of bidirectional update propagation between the source and the view.

\begin{figure}[t]
	\centering
	\begin{subfigure}[t]{0.195\textwidth}
		\begin{tikzpicture}[new set=import nodes,
		database/.style={
			cylinder,
			cylinder uses custom fill,
			shape border rotate=90,
			aspect=0.25 , 
			draw
		},
		setsize/.style={minimum width=0.6cm,minimum height=0.6cm}]
		\begin{scope}[nodes={set=import nodes}] 
		\node [draw,database, inner xsep=8pt]  (source) at (0,0) {\vspace{10pt}$S$\vspace{10pt}};
		\node [draw,rectangle,setsize]  (view) at (2,0) {$V$};
		\node [draw,database,  inner xsep=7pt]  (updatedsource) at (0,-1.5) {$S'$};
		\node [draw,rectangle, setsize]  (updatedview) at (2,-1.5) {$V'$};
		\node [draw=none]  (u) at (2-0.15,-0.7) {$u$};
		\node [draw=none]  (Tu) at (0+0.4,-0.7) {$T(u)$};
		\end{scope}
		\graph [edge quotes={inner sep=1pt, auto}, 
		grow down sep = 0.7cm, branch right sep= 1cm, nodes={draw,align=center}, edges=-stealth]  {
			(import nodes);
			source -> ["$get$"] view;
			source ->  updatedsource;
			view -> updatedview;
			u ->[dash dot] Tu;
		};
		\end{tikzpicture}
		\caption{}
		\label{fig:viewupdate}
	\end{subfigure}
	\begin{subfigure}[t]{0.195\textwidth}
		\begin{tikzpicture}[new set=import nodes,
		database/.style={
			cylinder,
			cylinder uses custom fill,
			shape border rotate=90,
			aspect=0.25,
			draw
		},
		setsize/.style={minimum width=0.6cm,minimum height=0.6cm}]
		\begin{scope}[nodes={set=import nodes}] 
		\node [draw,database, inner xsep=8pt]  (source) at (0,0) {$S$};
		\node [draw,rectangle,setsize]  (view) at (2,0) {$V$};
		\node [draw,database,  inner xsep=7pt]  (updatedsource) at (0,-1.5) {$S'$};
		\node [draw,rectangle, setsize]  (updatedview) at (2,-1.5) {$V'$};
		\node [draw=none]  (u) at (2-0.15,-0.7) {$u$};
		\end{scope}
		\graph [edge quotes={inner sep=1pt, auto}, 
		grow down sep = 0.7cm, branch right sep= 1cm, nodes={draw,align=center}, edges=-stealth]  {
			(import nodes);
			source -> ["$get$"] view;
			view -> updatedview;
			updatedview -> ["$put$"] updatedsource;
			source -> [bend left = 90] updatedsource;
		};
		\end{tikzpicture}
		\caption{}
		\label{fig:bx}
	\end{subfigure}
	\caption{The view update problem (a) and bidirectional transformation (b).}
	\label{fig:viewupdate_bx}
\end{figure}
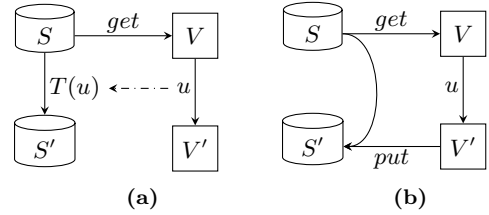

This idea is inspired by the research on bidirectional programming \cite{Foster2007,czarnecki2009} in the programming language community, where update propagation from the view to the source is formulated as a so-called \textit{putback transformation} $put$, which maps the updated view and the original source to an updated source, as shown in Figure~\ref{fig:bx}. This $put$ not only captures the view update strategy but also fully describes the view update behavior. First, it is clear that if we have such a putback transformation, the translation $T$ is obtained for free: 
\[
T(u)(S) = put (S,u(get(S))).
\]
Second, and more interestingly, while there may be many putback transformations for a view definition $get$, there is at most one view definition for a putback transformation $put$ for a well-behaved view update \cite{Hu2014, Fischer2015b, Fischer2015, Ko2016, KoHu2018}. Thus, $get$ can be deterministically derived from $put$ in general.
Although several languages have been proposed for writing $put$ for updatable views over tree-like data structures \cite{Zhu2016, Ko2016,KoHu2018}, whether we can design such a language for solving the classical view update problem on relations remains unclear.
 
There are several challenges in designing a formal language for programming $put$, a view update strategy, on relations.
\begin{itemize}
    \item The language is desired to be expressive in practice to cover users' update strategies.
    \item To make every view update consistent with the source database, an update strategy $put$ must satisfy some certain properties, as formalized in previous work \cite{Foster2007, Fischer2015, Fischer2015b}. Therefore, there is a need for a validation algorithm to statically check the well-behavedness of user-written strategies and whether they respect the view definition if the view is defined beforehand.
    \item To be useful in practice rather than just a theoretical framework, the language must be efficiently implemented when running in relational database management systems (RDBMSs).
\end{itemize}

In contrast to the existing approaches \cite{Zhu2016, Ko2016, KoHu2018} where new domain-specific languages (DSLs) are designed, we argue that Datalog, a well-known query language, can be used as a formal language for describing view update strategies in relational databases. Our contributions are summarized as follows.
\begin{itemize}
    \item We introduce a novel way to use nonrecursive Datalog with negation and built-in predicates for describing view update strategies. We propose a validation algorithm for statically checking the well-behavedness of the described update strategies.
    \item We identify a fragment of Datalog, called linear-view guarded negation Datalog (LVGN-Datalog), in which our validation algorithm is both sound and complete. Furthermore, the algorithm can automatically derive from view update strategies the corresponding view definition to confirm the view expected beforehand.
    \item We develop an incrementalization algorithm to optimize view update strategy programs. This algorithm integrates the standard incrementalization method for Datalog with the well-behavedness in view update.
    \item We have implemented all the algorithms in our framework, called BIRDS\footnote{A prototype implementation is available at \url{https://dangtv.github.io/BIRDS/}.}. 
    The experiments on benchmarks collected in practice show that our framework is feasible for checking most of the view update strategies. Interestingly, LVGN-Datalog is expressive enough for solving many types of views and can be efficiently implemented by incrementalization in existing RDBMSs.
\end{itemize}

The remainder of this paper is organized as follows.
After presenting some basic notions in Section~\ref{sec:background},
 we present our proposed method for specifying view update strategies in Datalog in Section~\ref{sec:LVGN-Datalog}.
The validation and incrementalization algorithms for these update strategies are described in Section~\ref{sec:verification} and Section~\ref{sec:incrementalization}, respectively.
Section~\ref{sec:experiment} shows the experimental results of our implementation.
Section~\ref{sec:relatedwork} summarizes related works.
Section~\ref{sec:conclusion} concludes this paper.

\section{Preliminaries}
\label{sec:background}
In this section, we briefly review the basic concepts and notations that will be used throughout this paper.
\subsection{Datalog and Relational Databases}
\textbf{Relational databases.}
A database schema $\mathcal{D}$ is a finite sequence of relation names (or predicate symbols, or simply predicates) $⟨r_1, \ldots , r_n⟩$.
Each predicate $r_i$ has an associated arity $n_i > 0$ or an associated sequence of attribute names $A_1, \ldots, A_{n_i}$. 
A database (instance) $D$ of $\mathcal{D}$ assigns to each predicate $r_i$ in $\mathcal{D}$ a finite $n_i$-ary relation $R_i$, $D(r_i) = R_i$. 

An atom (or atomic formula) is of the form $r( t_1 , \ldots , t_k)$ (or written as $r( \vec{t})$) such that $r$ is a $k$-ary predicate and each $t_i$ is a term, which is either a constant or a variable. When $t_1 , \ldots , t_k$ are all constants, $r( t_1 , \ldots , t_k)$ is called a ground atom.

A database $D$ can be represented as a set of ground atoms \cite{Ceri1989, Cali2012}, where each ground atom $r( t_1 , \ldots , t_k)$ corresponds to the tuple $\langle t_1 , \ldots , t_k \rangle$ of relation $R$ in $D$.
As an example of a relational database, consider a database $D$ that consists of two relations with respective schemas $r_1(A,B)$ and $r_2(C)$. Let the actual instances of these two relations be $R_1=\{\langle 1,2 \rangle, \langle 2,3 \rangle\}$ and $R_2=\{\langle 3 \rangle, \langle 4 \rangle\}$, respectively. The set of ground atoms of the database is $D = \{ r_1(1,2), r_1(2,3), r_2(3), r_2(4) \}$.

\textbf{Datalog.}
A Datalog program $P$ is a nonempty finite set of rules, and each rule is an expression of the form \cite{Ceri1989}:
\[ H \ruleeq L_1, \ldots, L_n. \]
where $H, L_1, \ldots, L_n$ are atoms. $H$ is called the rule head, and $L_1,\ldots,L_n$ is called the rule body.
The input of $P$ is a set of ground atoms, called the extensional database (EDB), physically stored in a relational database. The output of $P$ is all ground atoms derived through the program P and the EDB, called the intensional database (IDB).
Predicates in $P$ are divided into two categories: the EDB predicates occurring in the extensional database, and the IDB predicates occurring in the intensional database. An EDB predicate can never be the head predicate of a rule. The head predicate of each rule is an IDB predicate. We assume that each EDB/IDB predicate $r$ corresponds to exactly one EDB/IDB relation $R$.
Following the convention used in \cite{Ceri1989}, throughout this paper, we use lowercase characters for predicate symbols and uppercase characters for variables in Datalog programs. In a Datalog rule, variables that occur exactly once can be replaced by an anonymous variable, denoted as ``\_''.

A Datalog program $P$ can have many IDB predicates.
If restricting the output of $P$ to an IDB relation $R$ corresponding to IDB predicate $r$, we have a Datalog query, denoted as $(P, R)$. We say that an IDB predicate $r$ (or a query $(P,R)$) is satisfiable if there exists a database $D$ such that the IDB relation $R$ in the output of $P$ over $D$ is nonempty \cite{alicebook}.

We can extend Datalog by allowing negation and built-in predicates, such as equality ($=$) or comparison ($<,>$), in Datalog rule bodies but in a safe way in which each variable occurring in the negated atoms or the built-in predicates must also occur in some positive atoms \cite{Ceri1989}.

\subsection{Bidirectional Transformations}
A bidirectional transformation (BX) \cite{Foster2007} is a pair of a forward transformation $get$ and a backward (putback) transformation $put$, as shown in Figure~\ref{fig:bx}.
The forward transformation $get$ is a query over a source database $S$ that results in a view relation $V$. The putback transformation $put$ takes as input the original database $S$ and an updated view $V'$ to produce a new database $S'$.
To ensure consistency between the source database and the view, a BX must satisfy the following \textit{round-tripping} properties, called \textsc{GetPut} and \textsc{PutGet}:
\begin{align}
    \forall \; S, \qquad put  \left( S,\; get(S) \right) &= S \tag{\textsc{GetPut}} \label{eq:getput} \\
        \forall \; S, \; V', \qquad get \left(put  \left( S, \;V'\right) \right) &= V'  \tag{\textsc{PutGet}} \label{eq:putget}
\end{align}
The \ref{eq:getput} property ensures that unchanged views correspond to unchanged sources, while the \textsc{PutGet} property ensures that all view updates are completely reflected to the source such that the updated view can be computed again from the query $get$ over the updated source.

\begin{definition}[Validity of Update Strategy]~\newline
    A view update strategy $put$ is said to be valid if there exists a view definition $get$ such that $put$ and $get$ satisfy both \textsc{GetPut} and \textsc{PutGet}.
    \label{def:put_validity}
\end{definition}

The important property that makes putback essential for BXs is that a valid view update strategy $put$ uniquely determines the view definition $get$, which satisfies \textsc{GetPut} and \textsc{PutGet} with $put$. Therefore, although $put$ is written in a unidirectional (backward) manner, if $put$ is valid, it can capture both forward and backward directions. We state the uniqueness of the view definition $get$ in the following theorem, and the proof can be found in \cite{Fischer2015}.

\begin{theorem}[Uniqueness of View Definition]~\newline
    Given a view update strategy $put$, there is at most one view definition $get$ that satisfies \textsc{GetPut} and \textsc{PutGet} with $put$.
    \label{thm:unique_query}
\end{theorem}

\section{The Language for View Update Strategies}
\label{sec:LVGN-Datalog}
As mentioned in the introduction, it may be surprising that the base language that we are using for view update strategies is nonrecursive Datalog with negation and built-in predicates (e.g., $=$, $\neq$, $<$, $>$) \cite{Ceri1989}.
One might wonder how the pure query language Datalog can be used to describe updates.
In this section, we show that delta relations enable Datalog to describe view update strategies.
We will define a fragment of Datalog, called LVGN-Datalog, which is not only powerful for describing various view update strategies but also important for our later validation.

\subsection{Formulating Update Strategies as Queries Producing Delta Relations}
\begin{figure}[t]
    \centering
    \begin{tikzpicture}[new set=import nodes,
        database/.style={
            cylinder,
            cylinder uses custom fill,
            shape border rotate=90,
            aspect=0.25,
            draw
        },
        setsize/.style={minimum width=0.6cm,minimum height=0.6cm}]
        \begin{scope}[nodes={set=import nodes}] 
        \node [draw,database,  inner xsep=7pt]  (updatedsource) at (0-0.4,0) {$S'$};
        \node [draw,database, inner xsep=8pt]  (source) at (2,0) {$S$};
        \node [inner sep = 1pt] (apply) at (1-0.2,0) {$\bigoplus$};
        \node [font=\small] (applytext) at (1-0.2,0.4) {Apply};
        \node [draw,rectangle,setsize]  (view) at (4-0.5,0) {$V$};
        \node [draw,rectangle, database,  inner xsep=4.5pt]  (delta) at (2,-1) {$\Delta S$};
        \end{scope}
        \graph [edge quotes={inner sep=1pt, auto}, 
        grow down sep = 0.7cm, branch right sep= 1cm, nodes={draw,align=center}, edges=-stealth]  {
            (import nodes);
            view -> ["$putdelta$", bend left = 34] delta;
            source -> [bend left = 90] delta;
            source -> apply;
            delta -> [bend left = 45] apply;
            apply -> updatedsource;
        };
        \end{tikzpicture}
    \caption{View update strategy $put$.}
    \label{fig:bxdatalog}
\end{figure}
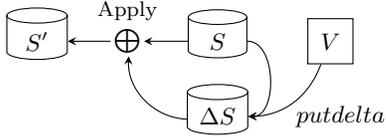
Recall that a view update strategy is a putback transformation $put$ that takes as input the original source database and an updated view to produce an updated source.
Our idea of specifying the transformation $put$ in Datalog is to write a Datalog query that takes as input the original source database and an updated view to yield updates on the source; thus, the new source can be obtained.

We use delta relations to represent updates to the source database. The concept of delta relations is not new and is used in the study on the incrementalization of Datalog programs \cite{Gupta1993}. Unlike the use of delta relations to describe incrementalization algorithms at the meta level, we let users consider both relations and their corresponding delta relations at the programming level.

Let $R$ be a relation and $r$ be the predicate corresponding to $R$. Following \cite{Green2015b, Koch2010, Fernando2011}, we use two delta predicates $+r$ and $-r$ and write $+r(\vec{t})$ and $-r(\vec{t})$ to denote the insertion and deletion of the tuple $\vec{t}$ into/from relation $R$, respectively.
An update that replaces tuple $\vec{t}$ with a new one $\vec{t'}$ is a combination of a deletion $-r(\vec{t})$ and an insertion $+r(\vec{t})$.
We use a delta relation, denoted as $\Delta R$, to capture both these deletions and insertions.
For example, consider a binary relation $R= \{ \langle 1,2 \rangle, \langle 1,3 \rangle\}$; applying a delta relation $\Delta R = \{ -r(1,2), +r(1,1)\}$ to $R$ results in $R'= \{ \langle 1,1 \rangle, \langle 1,3 \rangle\}$.
Let $\deltaAdd{R}$ be the set of insertions and $\deltaDel{R}$ be the set of deletions in $\Delta R$.
Applying $\Delta R$ to the relation $R$ is to delete tuples in $\deltaDel{R}$ from $R$ and insert tuples in $\deltaAdd{R}$ into $R$.
Considering set semantics, the delta application is the following:
\[R'= R \oplus \Delta R =( R \setminus \deltaDel{R}) \cup  \deltaAdd{R} \]
 An update strategy for a view can now be specified by a set of Datalog rules that define delta relations of the source database from the updated view. 
\begin{example} 
    Consider a source database $S$, which consists of two base relations, $R_1$ and $R_2$, with respective schemas $r_1(A)$ and $r_2(A)$, and a view relation $V$ defined by a union over $R_1$ and $R_2$: $V = get(S) = R_1 \cup R_2$. 
    To illustrate the ambiguity of updates to $V$, consider an attempt to insert a tuple $\langle 3 \rangle$ into the view $V$. There are three simple ways to update the source database: (i) insert tuple $\langle 3 \rangle$ into $R_1$, (ii) insert tuple $\langle 3 \rangle$ into $R_2$, and (iii) insert tuple $\langle 3 \rangle$ into both $R_1$ and $R_2$. 
    Therefore, the update strategy for the view needs to be explicitly specified to resolve the ambiguity of view updates. 
    Given original source relations $R_1$ and $R_2$ and an updated view relation $V$, the following Datalog program is one strategy for propagating data in the updated view to the source: 
    \begin{align*}
        {-r_1}(X) &\ruleeq r_1(X), \neg v(X). \\
        {-r_2}(X) &\ruleeq r_2(X), \neg v(X).\\
        {+r_1}(X) &\ruleeq v(X), \neg r_1(X), \neg r_2(X).
    \end{align*}
The first two rules state that if a tuple $\langle X \rangle$ is in $R_1$ or $R_2$ but not in $V$, it will be deleted from $R_1$ or $R_2$, respectively. The last rule states that if a tuple $\langle X \rangle$ is in $V$ but in neither $R_1$ nor $R_2$, it will be inserted into $R_1$.
Let the actual instances of the source and the updated view be $S = \{r_1(1), r_2(2), r_2(4)\}$ and $V = \{v(1), v(3), v(4)\}$, respectively. The input for the Datalog program is a database of both the source and the view $(S,V) = \{r_1(1), r_2(2), r_2(4), v(1), v(3), v(4)\}$. Thus, the result is delta relations $\Delta R_1 = \{ +r_1(3) \}$ and $\Delta R_2 = \{-r_2(2) \}$. By applying these delta relations to $S$, we obtain a new source database $S' = \{r_1(1), r_1(3), r_2(4)\}$.
\label{exp:basics}
\end{example}

Formally, consider a database schema $\mathcal{S} = ⟨r_1, \ldots , r_n⟩$ and a single view $v$. Let $S$ be a source database and $V$ be an updated view relation. We use $\Delta S$ to denote all insertions and deletions of all relations in $S$. For example, the $\Delta S$ in Example~\ref{exp:basics} is $\Delta S = \{ +r_1(3), -r_2(2) \}$. We say that $\Delta S$ is {\em non-contradictory} if it has no insertion/deletion of the same tuple into/from the same relation.
Applying a non-contradictory $\Delta S$ to a database $S$, denoted as $S \oplus \Delta S$, is to apply each delta relation in $\Delta S$ to the corresponding relation in $S$.
We use the pair $(S,V)$ to denote the database instance $I$ over the schema $⟨r_1, \ldots , r_n, v⟩$ such that $I(r_i) = S(r_i)$ for each $i\in[1,n]$ and $I(v) = V$. A view update strategy $put$ is formulated by a Datalog query $putdelta$ over the database $(S,V)$ that results in a $\Delta S$ (shown in Figure~\ref{fig:bxdatalog}) as follows:
\begin{align}
    put (S, V) = S \oplus putdelta (S, V) \label{eq:put_def}
\end{align}
The Datalog program $putdelta$ is called a \textit{Datalog putback program} (or \textit{putback program} for short).
The result of $putdelta$, $\Delta S$, should be non-contradictory to be applicable to the original source database $S$.
\begin{definition}[Well-definedness]
A putback program is well defined if, for every source database $S$ and view relation $V$, the program results in a non-contradictory $\Delta S$.
\end{definition}

\subsection{LVGN-Datalog}
We have seen that nonrecursive Datalog with extensions including negation and built-in predicates can be used for specifying view update strategies.
We now focus on the extensions of Datalog in which the satisfiability of queries is decidable. This property plays an important role in guaranteeing that the validity of putback programs is decidable.
Specifically, we define a fragment of Datalog, LVGN-Datalog, which is an extension of nonrecursive guarded negation Datalog (GN-Datalog \cite{Barany2012}) with equalities, constants, comparisons \cite{Ceri1989} and linear view predicate.
This Datalog fragment allows not only for writing many practical view update strategies but also for decidable checking of validity later.

\subsubsection{Nonrecursive GN-Datalog with Equalities, Co\-nstants, and Comparisons}
We consider a restricted form of negation in Datalog, called GN-Datalog \cite{Barany2015,Barany2012}, in which we can decide the satisfiability of any queries. In this way, we define LVGN-Datalog as an extension of this GN-Datalog fragment without recursion as follows:
\begin{itemize}
    \item Equality is of the form $t_1 = t_2$, where $t_1$/$t_2$ is either a variable or a constant.
    \item Comparison predicates $<$ ($>$) on totally ordered domains in the form of $X < c$ ($X > c$), where $X$ is a variable and $c$ is a constant.
    \item Constants may freely be used in Datalog rule bodies or rule heads without restriction.
    \item Every rule is negation guarded \cite{Barany2012} such that for every atom $L$ (or equality, or comparison) occurring either in the rule head or negated in the rule body, the body must have a positive atom or equality, called a \textit{guard}, containing all variables occurring in $L$.
\end{itemize}

\begin{example}
    The following rule is negation guarded:
    \[h(X,Y,Z) \ruleeq \underbrace{r_1(X,Y,Z)}_{guard}, \neg \underbrace{Z=1}_{equality}, \neg r_2(X,Y,Z).\]
    because the negated atom $r_2(X,Y,Z)$, negated equality $\neg {Z=1}$ and the head atom $h(X,Y,Z)$ are all guarded since all variables $X$, $Y$, and $Z$ are in the positive atom $r_1(X,Y,Z)$.
\end{example}

\subsubsection{Linear View} 
As formally proven in \cite{Fischer2015b}, the putback transformation $put$ must be lossless (i.e., injective) with respect to the view relation. This means that all information in the view must be embedded in the updated source.
To enable tracking this behavior of putback programs in LVGN-Datalog, we introduce a restriction called \textit{linear view}, which controls the usage of the view in the programs. By linear view, we mean that the view is linearly used such that there is no self-join and projection on the view. Every program in LVGN-Datalog conforms to the linear view restriction defined as follows.
\begin{definition}[Linear view]
A Datalog putback pr\-ogram conforms to the linear view restriction if the view occurs only in the rules defining delta relations, and in each of these delta rules, there is at most one view atom and no anonymous variable ($\_$) occurs in the view atom.
\label{def:linear_view}
\end{definition}
\begin{example}
Given a source relation $R$ of arity 3 and a view relation $V$ of arity 2, consider the following rules of the delta relation $\Delta R$: 
\begin{align}
        & {-r}(X,Y,Z) \ruleeq r(X,Y,Z), \neg \underbrace{v(X,Y)}_{\text{linear~view}}. \tag{$rule_1$}\\
        & {-r}(X,Y,Z) \ruleeq r(X,Y,Z), \neg \underbrace{v(X,\_)}_{\text{projection}}. \tag{$rule_2$} \\
        & {+r}(X,Y,Z) \ruleeq \underbrace{v(X,Y), v(Y,Z)}_{\text{self-join}}, \neg r(X,Y,Z). \tag{$rule_3$}
    \end{align}
    $(rule_1)$ conforms to the linear view restriction because $v(X,Y)$ occurs once in the rule body, whereas $(rule_2)$ and $(rule_3)$ do not because there is an anonymous variable ($\_$) in the atom of $v$ in $(rule_2)$ and there is a self-join of $v$ in ($rule_3$). 
\end{example}

\subsubsection{Integrity Constraints}
Since an updatable view can be treated as a base table, it is natural to create constraints on the view.
Similar to the idea of negative constraints introduced in \cite{Cali2012}, we extend the rules in LVGN-Datalog by allowing a truth constant \textit{false} (denoted as $\bot$) in the rule head for expressing integrity constraints. 
The linear view restriction defined in Definition~\ref{def:linear_view} is also extended that the view predicate can also occur in the rules having $\bot$ in the head.
In this way, a constraint, called the {\em guarded negation constraint}, is of the form $\forall \vec{X}, \Phi(\vec{X}) \rightarrow \bot$, where $\Phi(\vec{X})$ is the conjunction of all atoms and negated atoms in the rule body and $\Phi(\vec{X})$ is a guarded negation formula. The universal quantifiers $\forall \vec{X}$ are omitted in Datalog rules.

\begin{example}
    Consider a view relation $v(X,Y,Z)$. To prevent any tuples having $Z>2$ in the view $v$, we can use the following constraint:
    $\bot \ruleeq  v(X,Y,Z), Z > 2.$
\end{example}

\subsubsection{Properties}
We say that a query $Q$ is satisfiable if there is an input database $D$ such that the result of $Q$ over $D$ is nonempty. The problem of determining whether a query in nonrecursive GN-Datalog is satisfiable is known to be decidable \cite{Barany2012}. 
It is not surprising that allowing equalities, constants and comparisons in nonrecursive GN-Datalog does not make the satisfiability problem undecidable since the same already holds for guarded negation in SQL \cite{Barany2012}. The idea is that we can transform such a GN-Datalog query into an equivalent guarded negation first-order (GNFO) formula whose satisfiability is decidable \cite{Barany2015}.
\begin{lemma}
    The query satisfiability problem is decidable for nonrecursive GN-Datalog with equalities, constants and comparisons.
    \label{lemma:decidability}
\end{lemma}

Given a set of guarded negation constraints $\Sigma$ and a query $Q$, we say that $Q$ is satisfiable under $\Sigma$ if there is an input database $D$ satisfying all constraints in $\Sigma$ such that the result of $Q$ over $D$ is nonempty.
\begin{theorem}
    The query satisfiability problem for nonrecursive GN-Datalog with equalities, constants and comparisons under a set of guarded negation constraints is decidable. 
    \label{thm:gn-constraints}
\end{theorem}

\subsection{A Case Study}
\label{sec:program_examples}

\begin{figure}[t]
        \centering
\begin{Datalog}[numbers=none,xleftmargin=0.8em, framexleftmargin=0em, frame=single, title=Base tables, linewidth=0.466\textwidth]
male($emp\_name$: string, $birth\_date$: date).
female($emp\_name$: string, $birth\_date$: date).
others($emp\_name$: string, $birth\_date$: date,                            $gender$: string).
ed($emp\_name$: string, $dept\_name$: string).
eed($emp\_name$: string, $dept\_name$: string).
\end{Datalog}
\vspace{-0.5em}
\begin{Datalog}[numbers=none,xleftmargin=0.8em, framexleftmargin=0em, frame=single, title=Views, linewidth=0.466\textwidth]
ced$(E, D)$             $\ruleeq$ ed$(E, D) , ¬$ eed$(E, D)$.
residents$(E, B, G)$     $\ruleeq$ others$(E, B, G)$.
residents$(E, B, $`F'$)$    $\ruleeq$ female$(E, B)$.
residents$(E, B, $`M'$)$    $\ruleeq$ male$(E, B)$.
residents1962$(E, B, G)$ $\ruleeq$ residents$(E, B, G),$ 
           $¬ B < $`1962-01-01'$, ¬ B > $`1962-12-31'.
employees$(E, B, G)$$\ruleeq$ residents$(E,B,G), $ ced$(E, D)$.
retired$(E)$       $\ruleeq$ residents$(E, B, G), ¬$ced$(E, \_)$.
\end{Datalog}
\vspace{-0.5em}
\caption{Database and view schema.}
\label{fig:examples}
\end{figure}

We consider a database of five base tables shown in Figure~\ref{fig:examples}. The base tables \texttt{male}, \texttt{female} and \texttt{others} contain personal information. Table \texttt{ed} has all historical departments of each person, while \texttt{eed} contains only former departments of each person. 
We illustrate how to use LVGN-Datalog to describe update strategies for the views defined in Figure~\ref{fig:examples}.

For the view \texttt{residents}, which contains all personal information, we use the attribute $gender$ to choose relevant base tables for propagating updated tuples in \texttt{residents}. More concretely, if there is a person in \texttt{residents} but not in any of the source tables \texttt{male}, \texttt{female} and \texttt{other}, we insert this person into the table corresponding to his/her $gender$. In contrast, we delete from the source tables the people who no longer appear in the view. The Datalog putback program for \texttt{residents} is the following:
\begin{Datalog}[numbers=none,xleftmargin=0em, framexleftmargin=0em]
+male$(E,B)$   $\ruleeq$ residents$(E,B,$`M'$),$ 
               $¬$ male$(E,B), ¬$ others$(E,B,$`M'$)$.
-male$(E,B)$   $\ruleeq$ male$(E,B), ¬$ residents$(E,B,$`M'$)$.
+female$(E,B)$ $\ruleeq$ residents$(E,B,G),$ $G = $`F'$, $ 
                $¬$ female$(E,B), ¬$ others$(E,B,G)$.
-female$(E,B)$ $\ruleeq$ female$(E,B), ¬$ residents$(E,B,$`F'$)$.
+others$(E,B,G)$ $\ruleeq$ residents$(E,B,G), ¬$ $G = $`M'$, $
                  $¬$ $G = $`F'$, ¬$ others$(E,B,G)$.
-others$(E,B,G)$ $\ruleeq$ others$(E,B,G),$ 
                  $¬$ residents$(E,B,G)$.
\end{Datalog}

The view \texttt{ced} contains information about the current departments of each employee. 
We express the following update strategy for propagating updated data in this view to the base tables \texttt{ed} and \texttt{eed}.
If a person is in a department according to \texttt{ed} but he/she is currently no longer in this department according to \texttt{ced}, this department becomes his/her previous department and thus needs to be added to \texttt{eed}.
If a person used to be in a department according to \texttt{eed} but he/she returned to this department according to \texttt{ced}, then this department of him/her needs to be removed from \texttt{eed}.
\begin{Datalog}[numbers=none,xleftmargin=0em, framexleftmargin=0em]
+ed$(E,D)$  $\ruleeq$ ced$(E,D),$ $¬$ ed$(E,D)$.
-eed$(E,D)$ $\ruleeq$ ced$(E,D),$ eed$(E,D)$.
+eed$(E,D)$ $\ruleeq$ ed$(E,D),$ $¬$ ced$(E,D),$ $¬$ eed$(E,D)$.
\end{Datalog}

The view \texttt{residents1962} is defined from the view \texttt{residen\-ts} such that \texttt{residents1962} contains all residents that have a birth date in 1962. Interestingly, because the view \texttt{residents} is now updatable, \texttt{residents} can be considered as the source relation of \texttt{residents1962}. Therefore, we can write an update strategy on \texttt{residents1962} for updating \texttt{residents} instead of updating the base tables \texttt{male}, \texttt{female} and \texttt{others} as follows:
\begin{Datalog}[numbers=none,xleftmargin=0em, framexleftmargin=0em]
% Constraints:
$⊥$ $\ruleeq$ residents1962$(E,B,G), B > $`1962-12-31'.
$⊥$ $\ruleeq$ residents1962$(E,B,G), B < $`1962-01-01'.
% Update rules:
+residents$(E,B,G) \ruleeq$ residents1962$(E,B,G),$ 
                      $¬$ residents$(E,B,G)$.
-residents$(E,B,G) \ruleeq$ residents$(E,B,G), $
                      $¬$ $B < $`1962-01-01'$,$ 
                      $¬$ $B > $`1962-12-31'$, $
                      $¬$ residents1962$(E,B,G)$.
\end{Datalog}
We define the constraints to guarantee that in the updated view \texttt{residents1962}, there is no tuple having a value of the attribute $birth\_date$ not in 1962. Any view updates that violate these constraints are rejected. In this way, our update strategy is to insert into the source table \texttt{residents} any new tuples appearing in \texttt{residents1962} but not yet in \texttt{residents}. On the other hand, we delete only tuples in \texttt{residents} having $birth\_date$ in 1962 if they no longer appear in \texttt{residents1962}.

The view \texttt{employees} contains residents who are employed, whereas \texttt{retired} contains residents who retired. Since \texttt{emplo\-yees} and \texttt{retired} are defined from two updatable views \texttt{residents} and \texttt{ced}, we can use \texttt{residents} and \texttt{ced} as the source relations to write an update strategy of \texttt{employees}:
\begin{Datalog}[numbers=none,xleftmargin=0em, framexleftmargin=0em]
% Constraints:
$\bot$ $\ruleeq$ employees$(E,B,G), ¬$ ced$(E,\_)$.
% Update rules:
+residents$(E,B,G) \ruleeq$ employees$(E,B,G),$ 
                      $¬$ residents$(E,B,G)$.
-residents$(E,B,G) \ruleeq$ residents$(E,B,G),$ 
                     ced$(E,\_), ¬$ employees$(E,B,G)$.
\end{Datalog}
Interestingly, in this strategy, we use a constraint to specify more complicated restrictions of updates on \texttt{employees}. The constraint implies that there must be no tuple $\langle E,B,G \rangle$ in the updated view \texttt{employees} having the value $E$ of the attribute $emp\_name$, which cannot be found in any tuples of \texttt{ced}. In other words, the constraint does not allow insertion into \texttt{employees} an actual new employee who is not mentioned in the source relation \texttt{ced}. The update strategy then reflects updates on the view \texttt{employees} to updates on the source \texttt{residents}.

For \texttt{retired}, we describe an update strategy to update the current employment status of residents as follows:
\begin{Datalog}[numbers=none,xleftmargin=0em, framexleftmargin=0em]
-ced$(E,D)$       $\ruleeq$ ced$(E,D),$ retired$(E)$.
+ced$(E,D)$       $\ruleeq$ residents$(E,\_,\_), ¬$ retired$(E), $
                    $¬$ ced$(E,\_) , D=$`unknown'.
+residents$(E,B,G)$ $\ruleeq$ retired$(E), G=$`unknown'$,$ 
            $¬$ residents$(E,\_,\_), B=$`00-00-00'.
\end{Datalog}


We have presented the formal way to describe view update strategies using Datalog.
In the next section, we will present our proposed validation algorithm for checking the validity of these update strategies.
In fact, if an update strategy specified in LVGN-Datalog is valid, the corresponding view definition can be automatically derived and expressed in nonrecursive GN-Datalog with equalities, constants and comparisons. For all the update strategies in our case study, the view definitions derived by our validation algorithm are the same as the expected ones in Figure~\ref{fig:examples}.

\section{Validation Algorithm}
\label{sec:verification}

As mentioned in Section~\ref{sec:background}, a view update strategy must be valid (Definition~\ref{def:put_validity}) to guarantee that every view update is well-behaved.
In this section, we present an algorithm for checking the validity of user-written view update strategies.

\subsection{Overview}
Checking the validity of a view update strategy based on Definition~\ref{def:put_validity} is challenging since it requires constructing a view definition satisfying \textit{both} the \textsc{GetPut} and \textsc{PutGet} properties. Instead, we shall propose another way for the validity check based on the following important fact.
\begin{lemma}
Given a valid view update strategy $put$, if a view definition $get$ satisfies \textsc{GetPut}, then $get$ must also satisfy \textsc{PutGet} with $put$.
    \label{lem:view_derivation}
\end{lemma}

\begin{figure}[t]
    \centering
    \includegraphics[trim={0 0.4cm 0 0 },clip, scale=0.8]{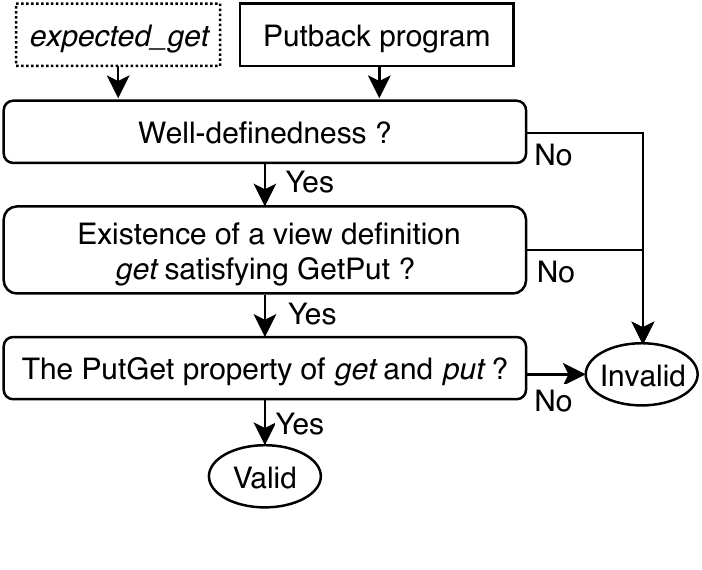}
    \caption{Validation algorithm.}
    \label{fig:validity-check}
\end{figure}

Lemma~\ref{lem:view_derivation} implies that if $put$ is valid, we can construct a view definition $get$ that satisfies both \ref{eq:getput} and \textsc{PutGet} by choosing any $get$ satisfying \ref{eq:getput}. 

By Lemma~\ref{lem:view_derivation}, the idea of our validation algorithm is detecting contradictions for the assumption that the given view update strategy $put$ is valid. Assuming that $put$ is valid, we first check the existence of a view definition $get$ satisfying \ref{eq:getput} with $put$. We consider the expected view definition $expected\_get$ if available as a candidate for the $get$ definition and construct the $get$ definition if $expected\_get$ does not satisfy \ref{eq:getput}. Clearly, if $get$ does not exist, we can conclude that $put$ is invalid. Otherwise, we continue to check whether $get$ also satisfies \textsc{PutGet} with $put$ (Lemma~\ref{lem:view_derivation}). If this check passed, we actually complete the validation and it is sufficient to conclude that $put$ is valid because the $get$ found satisfies both \textsc{GetPut} and \textsc{PutGet}. Furthermore, the constructed $get$ is useful to confirm the initially expected view definition especially when they are not the same. For the case in which the expected view definition is not explicitly specified, the view definition is automatically derived.

In particular, we are given a putback program $putdelta$, which is written in nonrecursive Datalog with negation and built-in predicates, and maybe an expected view definition ($expected\_get$) if it is explicitly described. The validation algorithm consists of three passes (see Figure~\ref{fig:validity-check}): (1) checking the well-definedness of the putback program, (2) checking the existence of a view definition $get$ satisfying \textsc{GetPut} with the view update strategy $put$ specified by the putback program and deriving $get$, and (3) checking whether $get$ and $put$ satisfy \textsc{PutGet}. If one of the passes fails, we can conclude that $put$ is invalid. Otherwise, $put$ is valid because the derived $get$ satisfies \textsc{GetPut} and \textsc{PutGet} with $put$.

\subsection{Well-definedness}
Consider a database schema $\mathcal{S} = \langle r_1, \ldots, r_n \rangle$ and a view $v$. Given a putback program $putdelta$, the goal is to check whether the delta $\Delta S$ resulting from $putdelta$ is non-contradi\-ctory for any source database $S$ and any view relation $V$. In other words, we check whether in $\Delta S$, there is no pair of insertion and deletion, ${+r_i}(\vec{t})$ and ${-r_i}(\vec{t})$, of the same tuple $\vec{t}$ on the same relation $R_i$.
To check this property, we add the following new rules to $putdelta$:
\begin{align}
    d_i(\vec{X_i}) \ruleeq {+r_i}(\vec{X_i}), {-r_i}(\vec{X_i}). \quad (i\in[1,n])
    \label{eq:disjoint_rule}
\end{align}
The problem of checking whether $\Delta S$ is non-contradictory is reduced to the problem of checking whether each IDB predicate $d_i$ in the Datalog program is unsatisfiable. When $putdelta$ is in LVGN-Datalog, because each rule~(\ref{eq:disjoint_rule}) is trivially negation guarded, according to Theorem~\ref{thm:gn-constraints}, the satisfiability of $d_i$ is decidable.

\subsection{Existence of A View Definition Satisfying GetPut}
Consider a view update strategy $put$ specified by a putback program $putdelta$ and a set of constraints $\Sigma$. Assume that $put$ is valid. If an expected view definition $expected\_get$ is explicitly written by users, we check whether $expected\_get$ satisfies \ref{eq:getput} with $put$. With the view defined by $expected\_get$, the \ref{eq:getput} property means that $put$ makes no change to the source. Therefore, checking the \ref{eq:getput} property is reduced to checking the unsatisfiability of each delta relation in the Datalog program $putdelta$. This check is decidable if $expected\_get$ and $putdelta$ are in LVGN-Datalog due to Theorem~\ref{thm:gn-constraints}.

If $expected\_get$ is not explicitly written or if it does not satisfy \ref{eq:getput}, we construct a view definition $get$ satisfying \ref{eq:getput} as follows. For each source database $S$, we find a steady-state view $V$ such that the putback transformation $put$ makes no change to the source database $S$. In other words, $V$ must satisfy the constraints in $\Sigma$ and $put (S, V) = S$. We define $get$ as the mapping that maps each $S$ to the $V$. If there exists an $S$ such that we cannot find any steady-state view, then there is no view definition satisfying \ref{eq:getput}, and we conclude that $put$ is invalid. Otherwise, the constructed $get$ satisfies \ref{eq:getput} with $put$. Moreover, the view relation $V$ resulting from $get$ over $S$ always satisfies $\Sigma$.

\begin{example}[Intuition]
Consider the update strategy $put$ in Example~\ref{exp:basics}.
For an arbitrary source database instance $S$, the goal is to find a steady-state view $V$ such that $put (S, V) = S$, i.e., both of the source relations $R_1$ and $R_2$ are unchanged. Recall that the putback transformation $put$ is described by Datalog rules that compute delta relations of each source relation $R_1$ and $R_2$. For $R_1$, we compute $\deltaAdd{R_1}$ and $\deltaDel{R_1}$, which are the set of insertions and the set of deletions on $R_1$, respectively. $R_1$ is unchanged if all inserted tuples are already in $R_1$ and all deleted tuples are actually not in $R_1$. Similarly, for $R_2$, all tuples in $\deltaDel{R_2}$ must be not in $R_2$ (we do not have $\deltaAdd{R_2}$). This leads to the following:
\begin{equation}
        \begin{array}{c}
            \deltaDel{R_1} \cap R_1 = \emptyset\\
            \deltaDel{R_2} \cap R_2 = \emptyset\\
            \deltaAdd{R_1} \setminus R_1 = \emptyset
        \end{array}
        \label{eqn:basics}
    \end{equation}
    Let us transform each delta predicate $-r_1$, $-r_2$, and $+r_1$ in the Datalog program $putdelta$ to the form of relational calculus query \cite{alicebook}:
        $\varphi_{-r_1} = r_1(X) \wedge \neg v(X)$,
        $\varphi_{-r_2} = r_2(X) \wedge \neg v(X)$,
        $\varphi_{+r_1} = v(X) \wedge \neg r_1(X) \wedge \neg r_2(X)$.
     The constraint (\ref{eqn:basics}) is equivalent to the constraint that all the relational calculus queries $\varphi_{-r_1}(X) \wedge r_1(X)$, $\varphi_{-r_2}(X) \wedge r_2(X)$ and $\varphi_{+r_1}(X) \wedge \neg r_1(X)$ result in an empty set over the database $(S, V)$ of both the source and view relations. In other words, $(S, V)$ does not satisfy the following first-order sentences:
    \[
        \left\{ \begin{array}{l}
            (S,V) \not\models  \exists X, \varphi_{-r_1}(X) \wedge r_1(X)  \\
            (S,V) \not\models  \exists X, \varphi_{-r_2}(X) \wedge r_2(X)  \\
            (S,V) \not\models \exists X, \varphi_{+r_1}(X) \wedge \neg r_1(X)
        \end{array}\right.
    \]
    By applying $\neg \exists X, \xi (X) \equiv \forall X, \xi (X) \rightarrow \bot$, we have
    \begin{align*}
        & \left\{   \begin{array}{l}
            (S,V) \models \forall X, \varphi_{-r_1}(X) \wedge r_1(X) \rightarrow \bot  \\
            (S,V) \models \forall X, \varphi_{-r_2}(X) \wedge r_2(X) \rightarrow \bot  \\
            (S,V) \models\forall X, \varphi_{+r_1}(X) \wedge \neg r_1(X)  \rightarrow \bot
        \end{array}\right.\\
        \Leftrightarrow& (S,V) \models \left\{ \begin{array}{l}
            \forall X, r_1(X) \wedge \neg v(X) \wedge r_1(X) \rightarrow \bot \\
            \forall X, r_2(X) \wedge \neg v(X) \wedge r_2(X) \rightarrow \bot \\
            \forall X, v(X) \wedge \neg r_1(X) \wedge \neg r_2(X) \wedge \neg r_1(X) \rightarrow \bot
        \end{array} \right. 
    \end{align*}

The idea for checking whether a view relation $V$ satisfying the above logical sentences exists is that we swap the atom $v(X)$ appearing in these sentences to either the right-hand side or the left-hand side of the implication formula. For this purpose, we apply $p \wedge \neg q \rightarrow \bot \equiv p \rightarrow q$ and obtain:
\begin{align*}
    \Leftrightarrow& (S,V) \models \left\{ \begin{array}{l}
        \forall X,r_1(X)  \rightarrow v(X) \\
        \forall X,r_2(X) \rightarrow v(X) \\
        \forall X,v(X)  \rightarrow \neg (\neg r_1(X) \wedge \neg r_2(X))
    \end{array} \right.
\end{align*}
By combining all sentences that have $v(X)$ on the right-hand side and combining all sentences that have $v(X)$ on the left-hand side, we obtain:
\begin{equation}
        \label{eqn:getput_fo}
        (S,V) \models \left\{ \begin{array}{l}
            \forall X, r_1(X) \vee r_2(X) \rightarrow v(X) \\
            \forall X, v(X)  \rightarrow \neg (\neg r_1(X) \wedge \neg r_2(X))
        \end{array}
        \right.
    \end{equation}
    Note that $S$ is an instance over $\langle r_1, r_2 \rangle$ and $V$ is the view relation corresponding to predicate $v$.
    The first sentence provides us the lower bound $V_{min}$ of $V$, which is the result of a first-order (FO) query\footnote{A FO query $\psi$ over $D$ results in all tuples $\vec{t}$ s.t. $D\models\psi(\vec{t})$.} $\psi_1 = r_1(X) \vee r_2(X)$ over $S$.
    The second sentence provides us the upper bound $V_{max}$ of $V$, which is the result of the first-order query $\psi_2 = \neg (\neg r_1(X) \wedge \neg r_2(X))$ over $S$.
    In fact, for each $S$, all the $V$ such that $V_{min} \subseteq V \subseteq V_{max}$ satisfy (\ref{eqn:getput_fo}), i.e., are steady-state instances of the view.
    Thus, a steady-state instance $V$ exists if $V_{min} \subseteq V_{max}$. 
    Indeed, by applying equivalence $\neg(p \vee q) \equiv \neg p \wedge \neg q$ to $\psi_2$, we obtain the same formula as $\psi_1$; hence, $\forall X,\psi_1(X) \to \psi_2(X)$ holds, leading to that $V_{min} \subseteq V_{max}$ holds. 
    Now by choosing $V_{min}$ as a steady-state view instance, we can construct a $get$ as the mapping that maps each $S$ to $V_{min}$. In other words, $get$ is a query equivalent to the FO query $\psi_1$ over the source $S$.
    Since $\psi_1$ is a safe-range formula\footnote{$\psi$ is a safe-range FO formula if all the variables in $\psi$ are range restricted \cite{alicebook}.}, we transform $\psi_1$ to an equivalent Datalog query\footnote{Due to the equivalence between nonrecursive Datalog queries and safe-range FO formulas \cite{alicebook}.} as follows:
    \begin{align}
        v(X) &\ruleeq r_1(X).\\
        v(X) &\ruleeq r_2(X).
    \end{align}
    This is the view definition $get$ that satisfies \ref{eq:getput} with the given view update strategy $put$.
    \label{exp:basic_view_derivation}
\end{example}

\subsubsection{Checking the existence of a steady-state view}
In general, similar to the idea shown in Example~\ref{exp:basic_view_derivation}, for an arbitrary putback program $putdelta$ and a set of constraints $\Sigma$ in LVGN-Datalog, we can always construct a guarded negation first-order (GNFO) sentence to check whether a steady-state view $V$ satisfying $\Sigma$ and $put(S,V) = S$ (i.e., $S \oplus putdelta (S,V) = S$) exists.

\begin{lemma}
    Given a LVGN-Datalog putback program putdelta and a set of guarded negation constraints $\Sigma$, there exist first-order formulas $\phi_1, \phi_2, \phi_3$ such that for a given database instance $S$, a view relation $V$ satisfies $\Sigma$ and $S \oplus putdelta (S,V) = S$ iff
    \begin{align}
        \left\{ \begin{array}{l}
            (S,V) \models \forall \vec{Y}, v(\vec{Y}) \wedge \phi_1(\vec{Y}) \rightarrow \bot\\
            (S,V) \models \forall \vec{Y}, \neg v(\vec{Y})  \wedge \phi_2(\vec{Y}) \rightarrow \bot\\
            (S,V) \not\models \phi_3
        \end{array}
        \right.
        \label{eq:sentence_combination}
    \end{align}
    where $v$ is the predicate corresponding to the view relation $V$ and $\phi_1, \phi_2, \phi_3$ have no occurrence of the view predicate $v$. Both $\phi_2(\vec{Y})$ and $\phi_3$ are safe-range GNFO formulas, and $v(\vec{Y}) \wedge \phi_1(\vec{Y})$ is equivalent to a GNFO formula.
    \label{lemma:solving_getput}
\end{lemma}

The third constraint $(S,V) \not\models \phi_3$ in (\ref{eq:sentence_combination}) is simplified to $S \not\models \phi_3$ because the FO sentence $\phi_3$ has no atom of $v$ as a subformula. This means that $\phi_3$ must be unsatisfiable over any database $S$. Since $\phi_3$ is a GNFO sentence, we can check whether $\phi_3$ is satisfiable. If it is satisfiable, we conclude that the view relation $V$ does not exist; thus, $put$ is invalid.

For the two other constraints in (\ref{eq:sentence_combination}), by applying the logical equivalence $p∧¬q→⊥≡p→q$, we have:
\begin{align}
    \left\{ \begin{array}{l}
        (S,V) \models \forall \vec{Y}, v(\vec{Y})  \rightarrow \neg \phi_1(\vec{Y})\\
        (S,V) \models \forall \vec{Y}, \phi_2(\vec{Y}) \rightarrow v(\vec{Y})
    \end{array}
    \right. \label{eq:view_derivation}
\end{align}
Because $\phi_1$ and $\phi_2$ do not contain an atom of $v$ as a subformula, there exists an instance $V$ if
\begin{align*}
    & S \models \forall \vec{Y}, \phi_2(\vec{Y}) \rightarrow \neg \phi_1(\vec{Y})\\
    \Leftrightarrow & S \models \forall \vec{Y}, \phi_1(\vec{Y}) \wedge \phi_2(\vec{Y}) \rightarrow \bot
\end{align*}
This means that the sentence $\exists \vec{Y}, \phi_1(\vec{Y}) \wedge \phi_2(\vec{Y})$ is not satisfiable.
In this way, checking the existence of a $V$ is now reduced to checking the satisfiability of $\exists \vec{Y}, \phi_1(\vec{Y}) \wedge \phi_2(\vec{Y})$.
The idea of checking the satisfiability of $\exists \vec{Y}, \phi_1(\vec{Y}) \wedge \phi_2(\vec{Y})$ is to reduce this problem to that of a GNFO sentence. For this purpose, by introducing a fresh relation $r$ of an appropriate arity, we have the fact that $\exists \vec{Y}, \phi_1(\vec{Y}) \wedge \phi_2(\vec{Y})$ is satisfiable if and only if $\exists \vec{Y}, r(\vec{Y}) \wedge \phi_1(\vec{Y}) \wedge \phi_2(\vec{Y})$ is satisfiable. Because $v(\vec{Y}) \wedge \phi_1(\vec{Y})$ is equivalent to a GNFO formula, $r(\vec{Y}) \wedge \phi_1(\vec{Y})$ is also equivalent to a GNFO formula. 
On the other hand, $\phi_2(\vec{Y})$ is equivalent to a GNFO formula; hence, we can transform $\exists \vec{Y}, r(\vec{Y}) \wedge \phi_1(\vec{Y}) \wedge \phi_2(\vec{Y})$ into an equivalent GNFO sentence whose satisfiability is decidable \cite{Barany2015}.

\subsubsection{Constructing a view definition} 
If both $\phi_3$ and $\exists \vec{Y}, \phi_1(\vec{Y}) \wedge \phi_2(\vec{Y})$ are unsatisfiable, there exists a steady-state view $V$ satisfying $\Sigma$ such that $S \oplus putdelta (S,V) = S$ for each database $S$. One steady-state view $V$ is the one resulting from the FO formula $\phi_2$ over $S$. Indeed, such a $V$ satisfies (\ref{eq:view_derivation}); hence, it satisfies $\Sigma$ and $S \oplus putdelta (S,V) = S$. By choosing this steady-state view, we can construct a view definition $get$ as the Datalog query equivalent to $\phi_2$ because $\phi_2$ is a safe-range formula. The equivalence of safe-range first-order logic and Datalog was well studied in database theory \cite{alicebook,Barany2012}. We present the detailed transformation from safe-range FO formula to Datalog query in Ex~\ref{sec:fo-to-datalog}. Due to Lemma~\ref{lemma:solving_getput}, $\phi_2$ is also negation guarded and hence, $get$ is in nonrecursive GN-Datalog with equalities, constants and comparisons.

\subsection{The PutGet Property}
To check the \textsc{PutGet} property that $get(put(S,V)) = V$ for any $S$ and $V$, we first construct a Datalog query over database $(S,V)$ equivalent to the composition $get(put(S,V))$.
Recall that $put(S,V) = S \oplus putdelta(S,V)$. The result of $put(S,V)$ is a new source $S'$ obtained by applying $\Delta S$ computed from $putdelta$ to the original source $S$.
Let us use predicate $r_i^{new}$ for the new relation of predicate $r_i$ in $S$ after the update.
The result of applying a delta $\Delta S$ to the database $S$ is equivalent to the result of the following Datalog rules ($i\in [1,n]$):
\begin{Datalog}[numbers=none,xleftmargin=0em]
$r_i^{new}$($\vec{X_i}$) $\ruleeq$ $r_i$($\vec{X_i}$), $\neg$ -$r_i$($\vec{X_i}$).
$r_i^{new}$($\vec{X_i}$) $\ruleeq$ +$r_i$($\vec{X_i}$).
\end{Datalog}
By adding these rules to the Datalog putback program $putdelta$, we derive a new Datalog program, denoted as $newsource$, that results in a new source database. The result of $get(put(S,V))$ is the same as the result of the Datalog query $get$ over the new source database computed by the program $newsource$. Therefore, we can substitute each EDB predicate $r_i$ in the program $get$ with the new program $r_i^{new}$ and then merge the obtained program with the program $newsource$ to obtain a Datalog program, denoted as $putget$.
The result of $putget$ over $(S,V)$ is exactly the same as the result of $get(put(S,V))$.
For example, the Datalog program $putget$ for the view update strategy in Example~\ref{exp:basic_view_derivation} is:
\begin{Datalog}[numbers=none,xleftmargin=0em]
-$r_1(X)$   $\ruleeq$ $r_1(X)$, $\neg$ $v(X)$.
-$r_2(X)$   $\ruleeq$ $r_2(X)$, $\neg$ $v(X)$.
+$r_1(X)$   $\ruleeq$ $v(X)$, $\neg$ $r_1(X)$, $\neg$ $r_2(X)$.
$r_1^{new}(X)$  $\ruleeq$ $r_1(X)$, $\neg$ -$r_1(X)$. 
$r_1^{new}(X)$  $\ruleeq$ +$r_1(X)$. 
$r_2^{new}(X)$  $\ruleeq$ $r_2(X)$, $\neg$ -$r_2(X)$. 
$v^{new}(X)$  $\ruleeq$ $r_1^{new}(X)$.
$v^{new}(X)$  $\ruleeq$ $r_2^{new}(X)$.
\end{Datalog}

Checking the \textsc{PutGet} property is now reduced to checking whether the result of Datalog query $putget$ over database $(S,V)$ is the same as the view relation $V$.
By transforming $putget$ to the FO formula $\phi_{putget}(\vec{Y})$, we reduce checking the \textsc{PutGet} property to checking the satisfiability of the two following sentences:
\begin{align}
    \Phi_1 &= \exists \vec{Y}, \phi_{putget}(\vec{Y}) \wedge \neg v(\vec{Y}) \label{eq:putget_sentence1}\\
    \Phi_2 &= \exists \vec{Y}, v(\vec{Y}) \wedge \neg \phi_{putget}(\vec{Y}) \label{eq:putget_sentence2}
\end{align}
The \textsc{PutGet} property holds if and only if $\Phi_1$ and $\Phi_2$ are not satisfiable. Clearly, if $get$ and $putdelta$ are in LVGN-Datalog, $putget$ is also in LVGN-Datalog, leading to that $\phi_{putget}(\vec{Y})$ is a GNFO formula. Therefore, $\Phi_2$ is a GNFO sentence; hence, its satisfiability is decidable. $\Phi_1$ is satisfiable if and only if $\Phi_1' = \exists \vec{Y}, \phi_{putget}(\vec{Y}) \wedge r(\vec{Y}) \wedge \neg v(\vec{Y})$ is satisfiable, where $r$ is a fresh relation of an appropriate arity. Since $\Phi_1'$ is a guarded negation first-order sentence, its satisfiability is decidable, and thus the satisfiability of $\Phi_1$ is also decidable.

\subsection{Soundness and Completeness}
\label{subsec:sound-completeness}
\begin{algorithm}[t]
    \caption{\textsc{Validate}($expected\_get$, $putdelta$, $\Sigma$)}
    \label{algo:LVGN-Datalog-checking}
    \DontPrintSemicolon
    $get$ $\leftarrow$ null;\\
    \tcp{Checking the well-definedness of $putdelta$}
    \textbf{check} if all predicates $d_i$ ($i\in[1,n]$) in (\ref{eq:disjoint_rule}) are unsatisfiable under $\Sigma$; \\
    \If{$expected\_get$ \textbf{is not} null}{
    	\tcp{Checking if $expected\_get$ satisfies \ref{eq:getput}}
        \If{all delta relations of $putdelta$ are unsatisfiable under $\Sigma$ with the view defined by $expected\_get$}{
        	$get$ $\leftarrow$ $expected\_get$;
        }
    }
    \If{($expected\_get$ \textbf{is} null) \textbf{or} ($get$  \textbf{is} null)} {
    	\tcp{Constructing a $get$ satisfying \ref{eq:getput}}
	    \textbf{check} if $\phi_3$ in (\ref{eq:sentence_combination}) is unsatisfiable under $\Sigma$;\\
	    \textbf{check} if $\exists \vec{Y}, \phi_1(\vec{Y}) \wedge \phi_2(\vec{Y})$ ($\phi_1$ and $\phi_2$ in (\ref{eq:view_derivation})) is unsatisfiable under $\Sigma$; \\
	    \tcp{Constructing a $get$}
	    $get \leftarrow$ Translating FO formula $\phi_2$ in (\ref{eq:view_derivation}) to an equivalent Datalog query;
	}
	\tcp{Checking the \textsc{PutGet} property}
    \textbf{check} if $\Phi_1$ and $\Phi_2$ in (\ref{eq:putget_sentence1}) and (\ref{eq:putget_sentence2}) are unsatisfiable under $\Sigma$; \\
    \textbf{return} $get$;
\end{algorithm}
Algorithm~\ref{algo:LVGN-Datalog-checking} summarizes the validation of Datalog putback programs $putdelta$. After all the checks have passed, the corresponding view definition is returned and $putdelta$ is valid. 
For LVGN-Datalog in which the query satisfiability is decidable (Theorem~\ref{thm:gn-constraints}), Algorithm~\ref{algo:LVGN-Datalog-checking} is sound and complete.
\begin{theorem}[Soundness and Completeness]~
\begin{itemize}
    \item If a LVGN-Datalog putback program $putdelta$ passes all the checks in Algorithm \ref{algo:LVGN-Datalog-checking}, $putdelta$ is valid.

    \item Every valid LVGN-Datalog putback program $putdelta$ passes all the checks in Algorithm \ref{algo:LVGN-Datalog-checking}.
\end{itemize}
\end{theorem}

It is remarkable that if $putdelta$ is not in LVGN-Datalog, but in nonrecursive Datalog with unrestricted negation and built-in predicates, we can still perform the checks in the validation algorithm by feeding them to an automated theorem prover. 
Though, Algorithm~\ref{algo:LVGN-Datalog-checking} may not terminate and not successfully construct the view definition $get$ because of the undecidability problem \cite{alicebook, Shmueli1993}. 
Therefore, Algorithm~\ref{algo:LVGN-Datalog-checking} is sound for validating the pair of $putdelta$ and $expected\_get$ that once it terminates, we can conclude $putdelta$ is valid.

\section{Incrementalization}
\label{sec:incrementalization}
We have shown that an updatable view is defined by a valid $put$, which makes changes to the source to reflect view updates. However, when there is only a small update on the view, repeating the $put$ computation is not efficient. In this section, we further optimize the computation of the putback program by exploiting its well-behavedness and integrating it with the standard incrementalization method for Datalog.

Consider the steady state before a view update in which both the source and the view are unchanged; due to the \ref{eq:getput} property, a valid $putdelta$ results in a $\Delta S$ having no effect on the original source $S$: $S\oplus\Delta S = S$. 
This means that $\Delta S$ can be either an empty set or a nonempty set in which all deletions in $\Delta S$ are not yet in the original source $S$ and all insertions in $\Delta S$ are already in $S$.
If the view is updated by a delta $\Delta V$, there will be some changes to $\Delta S$, denoted as $\Delta^2 S$, that have effects on the original source  $S$.


\begin{example}
    Consider the database in Example~\ref{exp:basics}: $S = \{r_1(1), r_2(2), r_2(4)\}$. Let $\Delta S = \{+r_1(1), +r_2(2), -r_2(3)\}$ be a delta of $S$. Clearly, $S \oplus \Delta S = S$. Now, we change $\Delta S$ by a delta of $\Delta S$, denoted as $\Delta^2 S$, which includes a set of deletions to $\Delta S$, $\Delta^{2-} S =  \{+r_1(1), -r_2(3)\}$, and a set of insertions to $\Delta S$, $\Delta^{2+} S =  \{+r_1(3), -r_2(4)\}$. We obtain a new delta of $S$:
    \[\Delta S' = (\Delta S \setminus \Delta^{2-} S) \cup \Delta^{2+} S = \{ +r_1(3), +r_2(2), -r_2(4) \}\]
    and the new database $S' = S \oplus \Delta S' = \{r_1(1), r_1(3), r_2(2)\}$.
    In fact, we can also obtain the same $S'$ by applying only $\Delta^{2+} S$ directly to $S$: $S' = S \oplus \Delta^{2+} S$.
\end{example}

Intuitively, for each base relation $R_i$ in the source database $S$, we obtain the new $R'_i$ by applying to $R_i$ the delta relations $\deltaDel{R_i}$ and $\deltaAdd{R_i}$ from $\Delta S$. Because all the tuples in $\deltaDel{R_i}$ are not in $R_i$ and all the tuples in $\deltaAdd{R_i}$ are in $R_i$, if we remove some tuples from $\deltaDel{R_i}$ or $\deltaAdd{R_i}$, then the result $R'_i$ has no change. Only the tuples inserted into $\deltaDel{R_i}$ or $\deltaAdd{R_i}$ make some changes in $R'_i$.
Therefore, $S'$ can be obtained by applying to the original $S$ the part $\Delta^{2+} S$ of $\Delta^2 S$, i.e., $\Delta S'$ and $\Delta^{2+} S$ are interchangeable.

\begin{proposition}
    \label{prop:putdelta_incrementalization}
    Let $S$ be a database and $\Delta S$ be a non-contradictory delta of the database $S$ such that $S \oplus \Delta S = S$. Let $\Delta^2 S$ be a delta of $\Delta S$, and the following equation holds:
    \begin{equation*}
        S' = S \oplus \Delta S'  = S \oplus \Delta^{2+} S
    \end{equation*}
    where $\Delta S' = \Delta S \oplus \Delta^2 S$ and $\Delta^{2+} S$ is the set of new tuples inserted into $\Delta S$ by applying $\Delta^2 S$.
\end{proposition}

Proposition~\ref{prop:putdelta_incrementalization} is the key observation for deriving from $putdelta$ an incremental Datalog program $\partial put$ that computes $\Delta S$ more efficiently (Figure~\ref{fig:Incrementalizing}).
To derive $\partial put$, we first incrementalize the Datalog program $putdelta$ to obtain Datalog rules that compute $\Delta^2 S$ from the change $\Delta V$ on the view $V$. This step can be performed using classical incrementalization methods for Datalog \cite{Gupta1993}. We then use $\Delta^{2+} S$ in $\Delta^2 S$ as an instance of $\Delta S$ for applying to the source $S$.

\begin{figure}[t]
    \centering
    \begin{subfigure}[t]{0.23\textwidth}
        \begin{tikzpicture}[new set=import nodes,
            database/.style={
                cylinder,
                cylinder uses custom fill,
                shape border rotate=90,
                aspect=0.25,
                draw
            },
            setsize/.style={minimum width=0.6cm,minimum height=0.4cm}]
            \begin{scope}[nodes={set=import nodes}] 
            \node [draw,database, inner xsep=8-2pt]  (source) at (2,0) {$S$};
            \node [draw,rectangle,setsize]  (view) at (4-0.4,0+0.1) {$V'$};
            \node [draw,rectangle, database,  inner xsep=4.5-3pt]  (delta) at (0.6,-1+0.2) {$\Delta S'$};
            \node [draw,rectangle,rounded corners]  (putdelta) at (2 ,-1+0.2) {$putdelta$};
            \node [draw,dash dot,rectangle,rounded corners, minimum width=1.7cm,minimum height=1.6cm]  (container) at (2.1 ,-0.3) {};
            \end{scope}
            \graph [edge quotes={inner sep=1pt, auto}, 
            grow down sep = 0.7cm, branch right sep= 1cm, nodes={draw,align=center}, edges=-stealth]  {
                (import nodes);
                view -> [bend left = 35] putdelta;
                source -> putdelta;
                putdelta -> delta;
            };
        \end{tikzpicture}
        \caption{Original $putdelta$}
        \label{fig:bx_inc_1}
    \end{subfigure}
    \begin{subfigure}[t]{0.23\textwidth}
        \begin{tikzpicture}[new set=import nodes,
            database/.style={
                cylinder,
                cylinder uses custom fill,
                shape border rotate=90,
                aspect=0.25,
                draw
            },
            setsize/.style={minimum width=0.6cm,minimum height=0.4cm}]
            \begin{scope}[nodes={set=import nodes}] 
            \node [draw,database, inner xsep=8-2pt]  (source) at (2,0) {$S$};
            \node [draw,rectangle,setsize]  (view) at (4-0.4-0.2,0+0.2) {$V$};
            \node [draw,rectangle,setsize,inner xsep=1.5pt]  (deltaview) at (4-0.4-0.2,-1+0.2) {$\Delta V$};
            \node [draw,rectangle, database,  inner xsep=4.5-3pt]  (delta) at (0.8,-1+0.2) {$\Delta S'$};
            \node [draw,rectangle,rounded corners]  (putdelta) at (2 ,-1+0.2) {$\partial put$};
            \node [draw,dash dot,rectangle,rounded corners, minimum width=1.3cm,minimum height=1.6cm]  (container) at (2.1 ,-0.3) {};
            \end{scope}
            \graph [edge quotes={inner sep=1pt, auto}, 
            grow down sep = 0.7cm, branch right sep= 1cm, nodes={draw,align=center}, edges=-stealth]  {
                (import nodes);
                view -> [bend left = 36] putdelta;
                deltaview -> putdelta;
                source -> putdelta;
                putdelta -> delta;
            };
        \end{tikzpicture}
        \caption{Incremental $\partial put$}
        \label{fig:bx_inc_2}
    \end{subfigure}
    \caption{Incrementalization of $putdelta$.}
    \label{fig:Incrementalizing}
\end{figure}
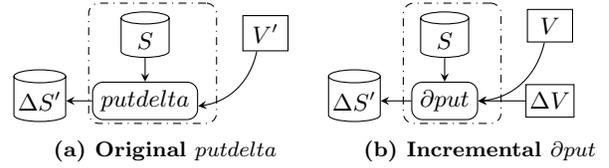

\begin{example}[Intuition]
    \label{exp:incrementalization}
    Given a source relation $R$ of arity 2 and a view relation $V$ defined by a selection on $R$:
    $v(X,Y) \ruleeq r(X,Y), Y>2.$
    Consider the following update strategy with a constraint that updates on $V$ must satisfy the selection condition $Y>2$:
    \[\begin{array}{rl}
            +r(X,Y) &\ruleeq v(X,Y), \neg r(X,Y).\\
            m(X,Y) &\ruleeq r(X,Y), Y>2. \\
            -r(X,Y) &\ruleeq m(X,Y), \neg v(X,Y).
    \end{array}\]
Let $\deltaAdd{V}$/$\deltaDel{V}$ be the set of insertions/deletions into/from the view $V$.
We use two predicates $+v$ and $-v$ for $\deltaAdd{V}$ and $\deltaDel{V}$, respectively.
To generate delta rules for computing changes of $\pm r$ when the view is changed by $\deltaAdd{V}$ and $\deltaDel{V}$, we adopt the incremental view maintenance techniques introduced in \cite{Gupta1993} but in a way that derives rules for computing the insertion set and deletion set for $\pm r$ separately.
When $\deltaAdd{V}$ and $\deltaDel{V}$ are disjoint, by applying distribution laws for the first Datalog rule, we derive two rules that define the changes to $\deltaAdd{R}$, a set of insertions $\Delta^+(\deltaAdd{R})$ and a set of deletions $\Delta^-(\deltaAdd{R})$, as follows:
\[\begin{array}{rl}
    +(+r)(X,Y) &\ruleeq +v(X,Y), \neg r(X,Y).\\
    -(+r)(X,Y) &\ruleeq -v(X,Y), \neg r(X,Y).
\end{array}\]
where predicates $+(+r)$ and $-(+r)$ correspond to $\Delta^+(\deltaAdd{R})$ and $\Delta^-(\deltaAdd{R})$, respectively.
Similarly, we derive rules defining changes to $\deltaDel{R}$, $\Delta^+(\deltaDel{R})$ and $\Delta^-(\deltaDel{R})$, as follows:
\[\begin{array}{rl}
    +(-r)(X,Y) &\ruleeq m(X,Y), -v(X,Y).\\
    -(-r)(X,Y) &\ruleeq m(X,Y), +v(X,Y).
\end{array}\]

Finally, as stated in Proposition~\ref{prop:putdelta_incrementalization}, $\Delta^{2+} S$ and $\Delta S'$ are interchangeable. Since $\Delta^{2+} S$ contains $\Delta^+(\deltaDel{R})$ and $\Delta^+(\deltaAdd{R})$, we can substitute $-r$ and $+r$ for the predicates $+(-r)$ and $+(+r)$, respectively, to derive the program $\partial put$ as follows:
\[\begin{array}{rl}
    m(X,Y) &\ruleeq r(X,Y), Y>2.\\
    +r(X,Y) &\ruleeq +v(X,Y), \neg r(X,Y). \\
    -r(X,Y) &\ruleeq m(X,Y), -v(X,Y).
\end{array}\]
Because $\deltaAdd{V}$ and $\deltaDel{V}$ are generally much smaller than the view $V$, the computation of $\Delta^+(\Delta^\pm_R)$ in the derived rules is more efficient than the computation of $\Delta^\pm_{R}$ in $putdelta$.
\end{example}

The incrementalization algorithm that transforms a putback program $putdelta$ in nonrecursive Datalog with negation and built-in predicates into an equivalent program $\partial put$ is as follows:
\begin{itemize}
	 \item \textit{Step 1}: We first stratify the Datalog program $putdelta$. Let $v, l_1,\ldots,l_m,\pm r_1, \ldots \pm r_n $ be a stratification \cite{Ceri1989} of the Datalog program $putdelta$, which is an order for the evaluation of IDB relations of $putdelta$. 
	 \item \textit{Step 2}: To derive rules for computing changes of each IDB relation $l_1,\ldots,l_m$ when the view $v$ is changed, we adopt the incremental view maintenance techniques introduced in \cite{Gupta1993} but in a way that derives rules for computing each insertion set ($+l_i$) and deletion set ($-l_i$) on IDB relation $l_i$ ($i\in [1,m]$) separately (see the details in Ex~\ref{sec:incrementalization_rewite_rules}).
	 \item \textit{Step 3}: Similar to \textit{Step 2}, we continue to derive rules for computing changes of each IDB relation $\pm r_1, \ldots \pm r_n $ but only for insertions to these relations. The purpose is to generate rules for computing $\Delta^{2+} S$, i.e., computing the relations $+(\pm r_1), \ldots +(\pm r_n) $.
	 \item \textit{Step 4}: We finally substitute $\pm r_i$ for $+(\pm r_i)$ ($i\in [1,n]$) in the derived rules to obtain the incremental program $\partial put$. This is because $\Delta^{2+} S$ can be used as an instance of $\Delta S'$ to apply to the source database $S$ (Proposition~\ref{prop:putdelta_incrementalization}).
\end{itemize}

As shown in Example~\ref{exp:incrementalization}, for a LVGN-Datalog program in which the view predicate $v$ occurs at most once in each delta rule, the transformation from a putback program $putdelta$ to an incremental one $\partial put$ is simplified to substituting $+v$ for positive predicate $v$ and  $-v$ for negative predicate $\neg v$.
\begin{lemma}
	\label{lem:LVGN-iincrementalization}
    Every valid LVGN-Datalog putback program $putdelta$ for a view relation $V$ is equivalent to an incremental program that is derived from $putdelta$ by substituting delta predicates of the view, $+v$ and $-v$, for positive and negative predicates of the view, $v$ and $\neg v$, respectively.
\end{lemma}

\section{Implementation and Evaluation}
\label{sec:experiment}

\subsection{Implementation}
We have implemented a prototype for our proposed validation and incrementalization algorithms in Ocaml (The full source code is available at \url{https://github.com/dangtv/BIRDS}). For the case in which the view update strategy is not in LVGN-Datalog, our framework feeds each check in our validation algorithm to the Z3 automated theorem prover \cite{z3}. As mentioned in Subsection~\ref{subsec:sound-completeness}, the validation algorithm may not terminate, though it is sound for checking the pair of view definition and update strategy program. We have also integrated our framework with PostgreSQL~\cite{postgresql}, a commercial RDBMS, by translating both the view definition and update strategy in Datalog to equivalent SQL and trigger programs.

Our translation is conducted because nonrecursive Datalog queries can be expressed in SQL \cite{alicebook}. We use a similar approach to the translation from Datalog to SQL used in \cite{Herrmann2017}.
The SQL view definition is of the form \texttt{CREATE VIEW <view-name> AS <sql-defining-query>}.
Meanwhile, the implementation for the update strategy is achieved by generating a SQL program that defines triggers \cite{Ramakrishnan1999} and associated trigger procedures on the view.
These trigger procedures are automatically invoked in response to view update requests, which can be
any SQL statements of \texttt{INSERT}/\texttt{DELETE}/\texttt{UPDATE}. Our framework also supports combining multiple SQL statements into one transaction to obtain a larger modification request on the view.
When there are view update requests, the triggers on the view perform the following steps: (1) handling update requests to the view to derive deltas of the view (see Ex~\ref{sec:deriveing_view_deltas}), (2) checking the constraints if applying the deltas from step (1) to the view, and (3) computing each delta relation and applying them to the source. The main trigger is as follows:
\begin{SQL}[numbers=none,xleftmargin=0em, framexleftmargin=0em]
CREATE TRIGGER <update-strategy> 
INSTEAD OF INSERT OR UPDATE OR DELETE ON <view $V$>
BEGIN
  -- Deriving changes on the view
  Derive $\deltaDel{V}$ and $\deltaAdd{V}$ from view update requests 
  -- Checking constraints
  FOR EACH <constraint $\forall \vec{X}, \Phi_i(\vec{X}) \to \bot$> DO
    IF EXISTS (<SQL-query-of $\Phi_i(\vec{X})$>) THEN 
      RAISE "Invalid view updates";
    END IF;
  END FOR;
  -- Calculating and applying delta relations
  FOR EACH <source relation $R_i$> DO
    CREATE TEMP TABLE $\deltaAdd{R_i}$ AS <sql-query-of $+r_i$>;
    CREATE TEMP TABLE $\deltaDel{R_i}$ AS <sql-query-of $-r_i$>;
    DELETE FROM $R_i$ WHERE ROW $(R_i)$ IN $\deltaDel{R_i}$; 
    INSERT INTO $R_i$ SELECT * FROM $\deltaAdd{R_i}$;
  END FOR;
END;
\end{SQL}


\begin{table*}[ht]
	\centering
	\caption{Validation results. S, P, SJ, IJ, LJ, RJ, FJ, U, D and A stand for selection, projection, semi join, inner join, left join, right join, full join, union, set difference and aggregation, respectively. PK, FK, ID, and C stand for primary key, foreign key, inclusion dependency, and domain constraint, respectively.}
	\small
	\begin{tabular}{|l|r|l|l| @{\hskip2pt}c@{\hskip2pt} |l| @{\hskip2pt}c@{\hskip2pt} |@{\hskip2pt}c@{\hskip2pt}|l|@{\hskip2pt}c@{\hskip2pt}|}
		\hline
		& ID & View            & \thead{Operator \\ in view \\ definition} & \thead{Program \\size \\(LOC)} & Constraint     & \thead{LVGN-\\Datalog} & \thead{NR-\\Datalog$^{\neg, =, <}$} & \thead{Validation \\Time (s)} & \thead{Compiled \\ SQL \\(Byte)} \\
		\hline \hline \multirow{23}{*}{\rotatebox[origin=c]{90}{Literature}} & 1  & car\_master     & P                                         & 4                              &            & \yes                   & \yes                 & 1.74               & 8447                          \\
		\cline{2-10}                                                         & 2  & goodstudents    & P,S                                       & 5                              & C          & \yes                   & \yes                 & 1.86               & 9182                          \\
		\cline{2-10}                                                         & 3  & luxuryitems     & S                                         & 5                              & C          & \yes                   & \yes                 & 1.77               & 8938                          \\
		\cline{2-10}                                                         & 4  & usa\_city       & P,S                                       & 5                              & C          & \yes                   & \yes                 & 1.77               & 9059                          \\
		\cline{2-10}                                                         & 5  & ced             & D                                         & 6                              &            & \yes                   & \yes                 & 1.72               & 8847                          \\
		\cline{2-10}                                                         & 6  & residents1962   & S                                         & 6                              & C          & \yes                   & \yes                 & 1.73               & 9699                          \\
		\cline{2-10}                                                         & 7  & employees       & SJ,P                                      & 6                              & ID         & \yes                   & \yes                 & 1.76               & 9358                          \\
		\cline{2-10}                                                         & 8  & researchers     & SJ,S,P                                    & 6                              &            & \yes                   & \yes                 & 1.79               & 9058                          \\
		\cline{2-10}                                                         & 9  & retired         & SJ,P,D                                    & 6                              &            & \yes                   & \yes                 & 1.76               & 9048                          \\
		\cline{2-10}                                                         & 10 & paramountmovies & P,S                                       & 7                              &            & \yes                   & \yes                 & 1.81               & 9721                          \\
		\cline{2-10}                                                         & 11 & officeinfo      & P                                         & 7                              &            & \yes                   & \yes                 & 1.8                & 9963                          \\
		\cline{2-10}                                                         & 12 & vw\_brands      & U,P                                       & 8                              & C          & \yes                   & \yes                 & 1.78               & 10932                         \\
		\cline{2-10}                                                         & 13 & tracks2         & P                                         & 8                              &            & \yes                   & \yes                 & 1.81               & 9824                          \\
		\cline{2-10}                                                         & 14 & residents       & U                                         & 10                             &            & \yes                   & \yes                 & 1.77               & 13504                         \\
		\cline{2-10}                                                         & 15 & tracks3         & S                                         & 11                             & C          & \yes                   & \yes                 & 1.88               & 14430                         \\
		\cline{2-10}                                                         & 16 & tracks1         & IJ                                        & 12                             & PK         & \no                    & \yes                 & 1.92               & 95606                         \\
		\cline{2-10}                                                         & 17 & bstudents       & IJ,P,S                                    & 13                             & PK         & \no                    & \yes                 & 2.13               & 22431                         \\
		\cline{2-10}                                                         & 18 & all\_cars       & IJ                                        & 13                             & PK, FK     & \no                    & \yes                 & 1.89               & 25013                         \\
		\cline{2-10}                                                         & 19 & measurement     & U                                         & 13                             & C, ID        & \yes                   & \yes                 & 1.78               & 12624                         \\
		\cline{2-10}                                                         & 20 & newpc           & IJ,P,S                                    & 15                             & JD         & \no                    & \yes                 & 2.06               & 44665                         \\
		\cline{2-10}                                                         & 21 & activestudents  & IJ,P,S                                    & 19                             & PK, JD     & \no                    & \yes                 & 2.19               & 31766                         \\
		\cline{2-10}                                                         & 22 & vw\_customers   & IJ,P                                      & 19                             & PK, FK, JD & \no                    & \yes                 & 2.92               & 26286                         \\
		\cline{2-10}                                                         & 23 & emp\_view       & IJ,P,A                                    & -                              &            & \no                    & \no                  & -                  & -                
		\\
		\hline \hline \multirow{9}{*}{\rotatebox[origin=c]{90}{Q\&A sites}} & 24 & ukaz\_lok         & S                                         & 6                              & C          & \yes                   & \yes                 & 1.79               & 10104                         \\
		\cline{2-10}                                                        & 25 & message           & U                                         & 8                              & C          & \yes                   & \yes                 & 1.8                & 15770                         \\
		\cline{2-10}                                                        & 26 & outstanding\_task & P, SJ                                     & 10                             & ID, C      & \yes                   & \yes                 & 10.07              & 18253                         \\
		\cline{2-10}                                                        & 27 & poi\_view         & P,IJ                                      & 12                             & PK         & \no                    & \yes                 & 2.1                & 24741                         \\
		\cline{2-10}                                                        & 28 & phonelist         & U                                         & 14                             & C          & \yes                   & \yes                 & 1.94               & 16553                         \\
		\cline{2-10}                                                        & 29 & products          & LJ                                        & 16                             & PK, FK, C  & \no                    & \yes                 & 3.6                & 58394                         \\
		\cline{2-10}                                                        & 30 & koncerty          & IJ                                        & 17                             & PK         & \no                    & \yes                 & 1.93               & 29147                         \\
		\cline{2-10}                                                        & 31 & purchaseview      & P,IJ                                      & 19                             & PK, FK, JD & \no                    & \yes                 & 1.89               & 27262                         \\
		\cline{2-10}                                                        & 32 & vehicle\_view     & P,IJ                                      & 20                             & PK, FK, JD & \no                    & \yes                 & 2.03               & 25226                              
		\\\hline
	\end{tabular}
	\label{table:compilation}     
\end{table*}
\subsection{Evaluation}
To evaluate our approach, we conduct two experiments. The goal of the first experiment is to investigate the practical relevance of our proposed method in describing view update strategies and to evaluate the performance of our framework in checking these described update strategies.
In the second experiment, we study the efficiency of our incrementalization algorithm when implementing updatable views in a commercial RDBMS.


\subsubsection{Benchmarks} 
To perform the evaluation, we collect benchmarks of views and update strategies from two different sources:
\begin{itemize}
    \item View update examples and exercises collected from the literature: textbooks \cite{Ramakrishnan1999, databasecompletebook}, online tutorials \cite{mysqltutorial,oracletutorial,postgresqltutorial,postgresqldocs,sqlservertutorial} (triggers, sharded tables, and so forth), papers \cite{Bohannon2006, Keller1985} and our case study in Section~\ref{sec:LVGN-Datalog}.
    \item View update issues asked on online question \& answer sites: Database Administrators Stack Exchange \cite{dba} and Stack Overflow Public Q\&A \cite{sof}.
\end{itemize}
All experiments on these benchmarks are run using Ubuntu server LTS 16.04 and PostgreSQL 9.6 on a computer with 2 CPUs and 4 GB RAM.

\subsubsection{Results} 
As mentioned previously, we perform the first experiment to investigate which users' update strategies are expressible and validatable by our approach. In our benchmarks, the collected view update strategies are either implemented in SQL triggers or naturally described by users/systems. We manually use nonrecursive Datalog with negation and built-in predicates (NR-Datalog$^{\neg,=,<}$) to specify these update strategies as $putdelta$ programs\footnote{For the update strategies implemented in SQL triggers, rewriting them into $putdelta$ programs can be automated.} and input them with the expected view definition to our framework.
Table~\ref{table:compilation} shows the validation results. In terms of expressiveness, NR-Datalog$^{\neg,=,<}$ can be used to formalize most of the view update strategies with many common integrity constraints except one update strategy for the aggregation view \texttt{emp\_view} (\#23). This is because we have not considered aggregation in Datalog. Interestingly, LVGN-Datalog can also express many update strategies for many views defined by selection, projection, union, set difference and semi join. Inner join views such as \texttt{all\_car} (\#18) are not expressible in LVGN-Datalog because the definition of inner join is not in guarded negation Datalog\footnote{An example of inner join is $v(X,Y,Z) \ruleeq s_1(X,Y),s_2(Y,Z)$, which is not a guarded negation Datalog rule.}.
LVGN-Datalog is also limited in expressing primary key (functional dependency) or join dependency because these dependencies are not negation guarded\footnote{Primary key $A$ on relation $r(A,B)$ is expressed by the rule $\bot \ruleeq r(A,B_1), r(A,B_2), \neg B_1 = B_2$, where the equality $B_1 = B_2$ is not guarded.}.
Even for the cases that LVGN-Datalog cannot express, thus far, all the well-behavedness checks in our experiment terminate after an acceptable time (approximately a few seconds). The validation time almost increases with the number of rules in the Datalog programs (program size), but this time also depends on the complexity of the source and view schema. For example, the update strategy of \texttt{outstanding\_task} (\#26) has the longest validation time because this view and its source relations have many more attributes than other views. Similarly, the size of the generated SQL program is larger for the more complex Datalog update strategies.


\begin{figure}[t]
    \centering
    \begin{subfigure}[t]{0.23\textwidth}
        \includegraphics[trim={0 0 0 0.3cm },clip, scale=0.54]{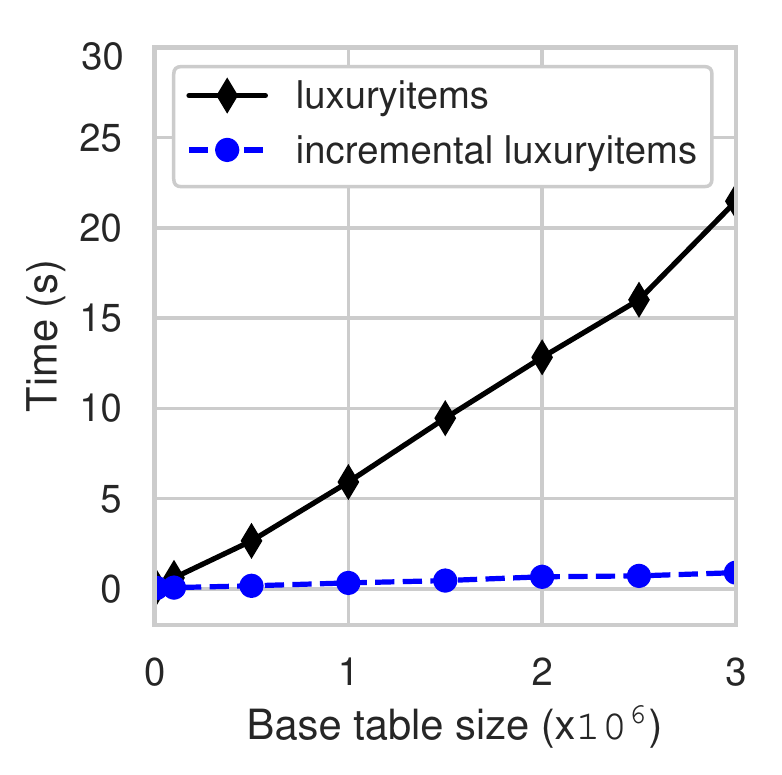}
        \caption{{luxuryitems}}
        \label{fig:luxuryitems_1tuple}
    \end{subfigure}
    \begin{subfigure}[t]{0.23\textwidth}
        \includegraphics[trim={0 0 0 0.3cm },clip, scale=0.54]{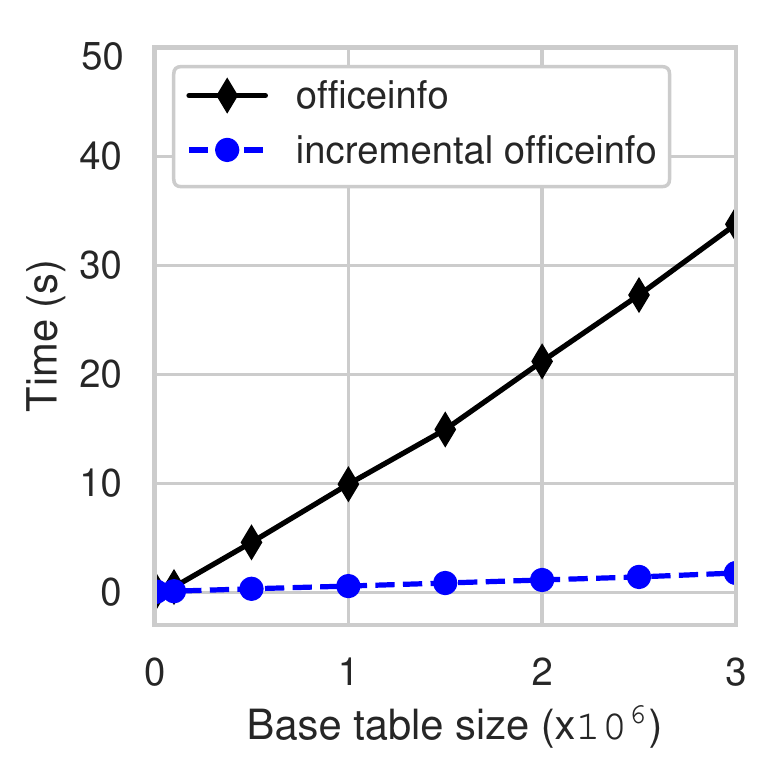}
        \caption{{officeinfo}}
        \label{fig:officeinfo_1tuple}
    \end{subfigure}
    \begin{subfigure}[t]{0.23\textwidth}
        \includegraphics[trim={0 0 0 0.3cm },clip, scale=0.54]{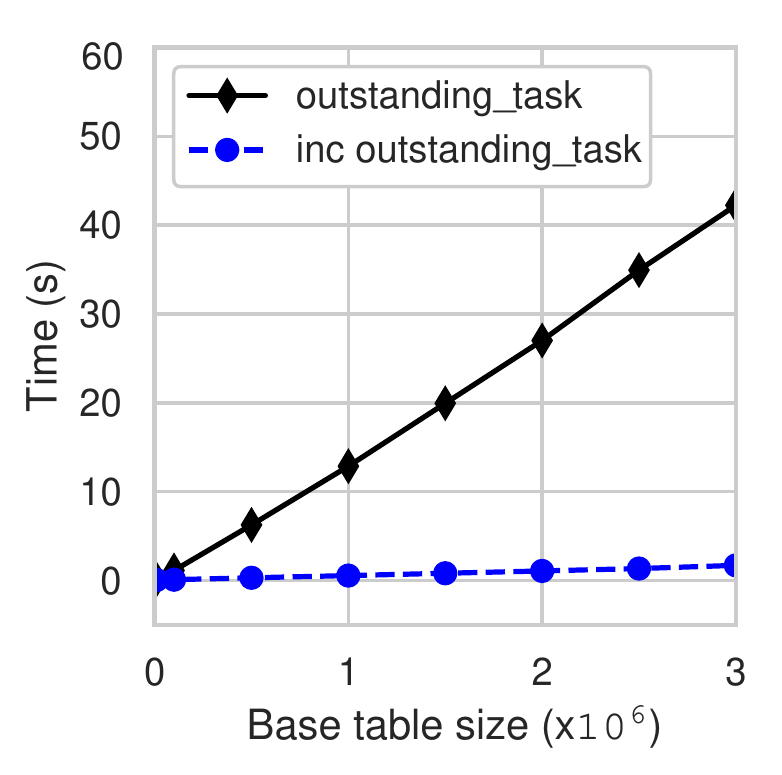}
        \caption{{outstanding\_task}}
        \label{fig:outstanding_task_1tuple}
    \end{subfigure}
    \begin{subfigure}[t]{0.23\textwidth}
        \includegraphics[trim={0 0 0 0.3cm },clip, scale=0.54]{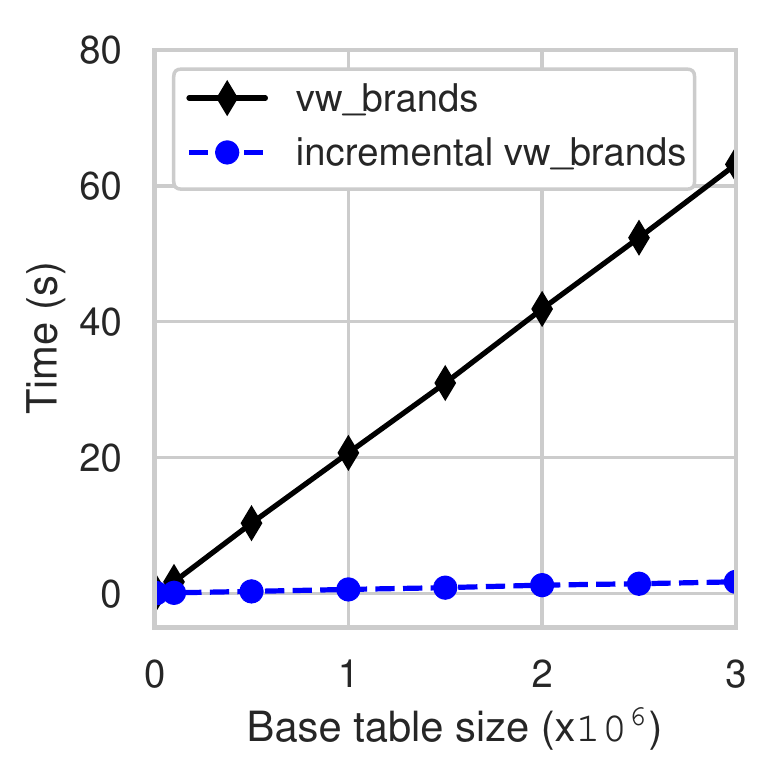}
        \caption{vw\_brands}
        \label{fig:vw_brands_1tuple}
    \end{subfigure}
    \caption{View updating time.}
    \label{fig:viewupdate_time}
\end{figure}

We perform the second experiment to evaluate the efficiency of the incrementalization algorithm in optimizing view update strategies. Specifically, we compare the performance of the incrementalized update strategy with the original one when they are translated into SQL trigger programs and run in the PostgreSQL database.
For this experiment, we select some typical views in our benchmarks including: \texttt{luxuryitems} (Selection), \texttt{officeinfo} (Projection), \texttt{outstanding\_task} (Join) and \texttt{vw\_brands} (Union).
For each view, we randomly generate data for the base tables and measure the running time of the view update strategy against the base table size (number of tuples) when there is an SQL statement that attempts to modify the view.
Figure~\ref{fig:viewupdate_time} shows the comparison between the original view update strategies (black lines) and the incrementalized ones (blue lines).
It is clear that as the size of the base tables increases, our incrementalization significantly reduces the running time to a constant value, thereby improving the performance of the view update strategies.

\section{Related work}
\label{sec:relatedwork}

The view update problem is a classical problem that has a long history in database research \cite{Fagin1983, Dayal1978, Dayal1982, Bancilhon1981, Keller1986, Medeiros1986, Keller1985, Kotidis2006, Herrmann2017, Buneman2002, Kimelfeld2012, Masunaga1984, Masunaga2017, Masunaga2018, Larson1991, Langerak1990}. It was realized very early that a database update that reflects a view update may not always exist, and even if it does exist, it may not be unique \cite{Dayal1978, Dayal1982}.
To solve the ambiguity of translating view updates to updates on base relations, the concept of view complement is proposed to determine the unique update translation of a view \cite{Bancilhon1981, Keller1984, Lechtenborger2003, Langerak1990}. Keller \cite{Keller1986} enumerates all view update translations and chooses the one through interaction with database administrators, thereby solving the ambiguity problem.
Some other researchers allow users to choose the one through an interaction with the user at view definition time \cite{Keller1986, Larson1991}. Some other approaches restrict the syntax for defining views \cite{Dayal1982} that allow for unambiguous update propagation.
Recently, intention-based approaches have been proposed to find relevant update policies for several types of views \cite{Masunaga1984, Masunaga2017, Masunaga2018}.
In another aspect, because some updates on views are not translatable, some works permit side effects of the view update translator \cite{Medeiros1986} or restrict the kind of updates that can be performed on a view \cite{Keller1985}. Some other works use auxiliary tables to store the updates, which cannot be applied to the underlying database \cite{Kotidis2006, Herrmann2017}.
The authors of \cite{Buneman2002, Kimelfeld2012} studied approximation algorithms to minimize the side effects for propagating deletion from the view to the source database.
However, these existing approaches can only solve a very restricted class of view updates.

By generalizing view update as a synchronization problem between two data structures, considerable research effort has been devoted to bidirectional programming \cite{czarnecki2009} for this problem not only in relational databases \cite{Bohannon2006, Horn2018} but also for other data types, such tree \cite{Foster2007, Matsuda2007}, graph \cite{Hidaka2010} or string data \cite{Barbosa2010}.
The prior work by Bohannon et al. \cite{Bohannon2006} employs bidirectional transformation for view update in relational databases. The authors propose a bidirectional language, called relational lenses, by enriching the SQL expression for defining views of projection, selection, and join. The language guarantees that every expression can be interpreted forwardly as a view definition and backwardly as an update strategy such that these backward and forward transformations are well-behaved. A recent work \cite{Horn2018} has shown that incrementalization is necessary for relational lenses to make this language practical in RDBMSs. However, this language is less expressive than general relational algebra; hence, not every updatable view can be written. Moreover, relational lenses still limit programmers from control over the update strategy.

Melnik et al. \cite{Melnik2008} propose a novel declarative mapping language for specifying the relationship between application entity views and relational databases, which is compiled into bidirectional views for the view update translation.
The user-specified mappings are validated to guarantee the generated bidirectional views to roundtrip. 
Furthermore, the authors introduce the concept of merge views that together with the bidirectional views contribute to determining complete update strategies, thereby solving the ambiguity of view updates. Though, merge views are exclusively used and validating the behavior of this operation with respect to the roundtripping criterion is not explicitly considered.
In comparison to \cite{Melnik2008}, where the proposed mapping language is restricted to selection-projection views (no joins), our approach focuses on a specification language, which is in the lower level but more expressive that more view update strategies can be expressed. Moreover, the full behaviour of the specified view update strategies is validated by our approach.

Our work was greatly inspired by the putback-based approach in bidirectional programming \cite{Hu2014, Pacheco2014a, Pacheco2014b, Fischer2015b, Ko2016, KoHu2018}. The key observation in this approach is that thanks to well-behavedness, putback transformation uniquely determines the get one.
In contrast to the other approaches, the putback-based approach provides languages that allow programmers to write their intended update strategies more freely and derive the $get$ behavior from their putback program. A typical language of this putback-based approach is BiGUL \cite{Ko2016, KoHu2018}, which supports programming putback functions declaratively while automatically deriving the corresponding unique forward transformation. Based on BiGUL, Zan et al. \cite{Zan2016} design a putback-based Haskell library for bidirectional transformations on relations. However, this language is designed for Haskell data structures; hence, it cannot run directly in database environments. The transformation from tables in relational databases to data structures in Haskell would reduce the performance of view updates. In contrast, we propose adopting the Datalog language for implementing view update strategies at the logical level, which will be optimized and translated to SQL statements to run efficiently inside an SQL database system.

\section{Conclusions}
\label{sec:conclusion}
In this paper, we have introduced a novel approach for relational view update in which programmers are given full control over deciding and implementing their view update strategies.
By using nonrecursive Datalog with extensions as the language for describing view update strategies, we propose algorithms for validating user-written update strategies and optimizing update strategies before compiling them into SQL scripts to run effectively in RDBMSs. The experimental results show the performance of our framework in terms of both validation time and running time.



\bibliographystyle{abbrv}
\bibliography{references}  



\begin{appendix}
\section{Proofs}

\subsection{Proof of Theorem \ref{thm:unique_query}}
\begin{proof}
    By contradiction. If there are two view definitions $get_1$ and $get_2$ that satisfy the condition, then by applying the \textsc{GetPut} and \textsc{PutGet} properties to the expression $E= get_1(put(S,get_2(S)))$, we have $E= get_1(S)$ and $E=get_2(S)$, respectively. This means that $get_1(S) = get_2(S)$ for any database $S$, i.e., $get_1$ and $get_2$ are equivalent.
\end{proof} 

\subsection{Proof of Lemma~\ref{lemma:decidability}}
\label{sec:proof-decidability}
\begin{proof}
	Let $P$ be a Datalog program in nonrecursive GN-Datalog with equalities, constants and comparisons.
	We shall transform a query $(P,R)$, where $R$ is an IDB relation corresponding to IDB predicate $r$ in $P$, into an equivalent guarded negation first-order (GNFO) formula \cite{Barany2015}.
	Without loss of generality, we assume that in $P$, for every pair of head atoms $h_1(\vec{X_1})$, $h_2(\vec{X_2})$ in $P$, $h_1 = h_2$ implies $\vec{X_1} =\vec{X_2} $ (this can be achieved by variable renaming).
	
	Since there are constants that can occur in both atoms and equalities. We first remove all constants appearing in atoms by transforming them to constants appearing in equalities. This can be done by introducing a fresh variable $X$ for each constant $c$ in the atoms of the Datalog rule (head or body), then adding an equality $X=c$ to the rule body and substitute $X$ for the constant $c$. By this transformation, we consider equalities of the form $X=c$ and a positive atom as a guard for negative predicates or head atom of Datalog rules. In other words, for each head atom or negative predicate $\beta$, there is a positive atom $p(\vec{Y})$ such that all the free variables in $\beta$ must appear in $p(\vec{Y})$ or in an equality of the form $X=c$.
	 For example, the following rule
	\[h(Z,1) \ruleeq p(Z,W, 3), \neg r(W, 4). \]
	is transformed into 
	\begin{align*}
		h(Z,X_1) \ruleeq & p(Z,W, X_2), \neg r(W, X_3), X_1 = 1, X_2 = 3, \\
							& X_3 = 4.
	\end{align*}
	in which the negated atom $r(W, X_3)$ is guarded by the positive atom $p(Z,W, X_2)$ and the equality $X_3 = 4$. The head atom $h(Z,X_1)$ is guarded by $p(Z,W, X_2)$ and $X_1 = 1$.
	
	We shall define a FO formula $\varphi_{r}$ equivalent to the Datalog query $(P,R)$, i.e, for every database $D$, the IDB relation $R$ (corresponding to IDB predicate $r$ in $P$) in the output of $P$ over $D$ (denoted as $P(D)|_R$) is the same as the set of tuple $\vec{t}$ satisfying $\varphi_{r}$ ($\{\vec{t} ~|~ \varphi(\vec{t})\}$).
	The construction of $\varphi_{r}$ is inductively defined as the following:
	\begin{itemize}
		\item (Base case) $r$ is an EDB predicate, i.e., $r\in \mathcal{S} \cup \{v\}$: $\varphi_{r} = r(\vec{X_r})$, where $\vec{X_r}$ denotes $(X_1,\ldots,X_{arity(r)})$.
		\item (Inductive case) $r$ is an IDB predicate, i.e., $r$ occurs in the head of some rules. Suppose that there are $m$ rules:
		\[ 
		\begin{array}{l}
		r(\vec{X_r}) \ruleeq \alpha_{1,1}, \ldots, \alpha_{1,n_1}.\\
		\ldots\\
		r(\vec{X_r}) \ruleeq \alpha_{m,1}, \ldots, \alpha_{m,n_m}.
		\end{array}
		\]
Let $\varphi_{r,i}(\vec{X_r})$ be the FO formula for $r$ when considering only the $i$-th rule:
\[\varphi_{r,i}(\vec{X_r}) = \exists \vec{E}_i, \bigwedge\limits_{j=1}^{n_i} \beta_{{i,j}}\]
where $\vec{E}_i$ contains the bound variables of the $i$-th rule (variables not in the rule head),
\[ \beta_{i,j} = \left\{
    \begin{array}{ll}
		\varphi_{w}(\vec{Z}),  \text{ if } \alpha_{i,j} \text{ is an atom } w(\vec{Z})\\
		\neg \varphi_{w}(\vec{Z}),  \text{ if } \alpha_{i,j} \text{ is a negated atom } \neg w(\vec{Z}) \\
		\alpha_{i,j} \text{ if } \alpha_{i,j} \text{ is an equality or a negated }\\
			\qquad\qquad\text{ equality }  \\
		C_{<c}(X) \text{ if } \alpha_{i,j} \text{ is a comparison predicate} \\
		\qquad \qquad X<c
		\\
		C_{>c}(X) \text{ if } \alpha_{i,j} \text{ is a comparison predicate} \\
		\qquad \qquad X>c
    \end{array}\right.
\]
		Here we introduce fresh predicates $C_{<c}(X)$ and $C_{>c}(X)$ for the comparisons.
		We have:
		\[\varphi_{r}(\vec{X_r}) = \bigvee\limits_{i = 1}^m \varphi_{r,i}(\vec{X_r}) = \bigvee\limits_{i = 1}^m \left( \exists \vec{E}_i, \bigwedge\limits_{j=1}^{n_i} \beta_{{i,j}}\right)\]
	\end{itemize}
	
	It is not difficult to show that $\varphi_{r}$ is equivalent to the Datalog query $(P,R)$. Indeed, for any database instance $D$, by induction, we can show that for each IDB predicate $r$ and each tuple $\vec{t}$, 
	\[ r(\vec{t}) \in P(D) \Leftrightarrow D \models \varphi_{r}(\vec{t})\]

	In each conjunction $\varphi_{r,i}(\vec{X_r}) = \exists \vec{E}_i, \bigwedge_{j=1}^{n_i} \beta_{{i,j}}$, each negative predicate $\beta_{i,j}$ is guarded by a positive atom $w_{i,j}(\vec{Y})$ and many equalities. Moreover, there exists a positive atom $w_i(\vec{Y})$ containing all the free variables of $\vec{X_r}$.

	Let us briefly recall the syntax of GNFO formulas with constants proposed by B\'{a}r\'{a}ny et al. \cite{Barany2015}. GNFO formulas with constants are generated by the following definition: \[φ::=r(t_1,\ldots,t_n)~|~t_1 =t_2 ~|~ φ_1 ∧φ_2 ~|~ φ_1 ∨φ_2 ~|~∃xφ ~|~ α∧¬φ\]
	where each $t_i$ is either a variable or a constant symbol, and, in $α∧¬φ$, $α$ is an atomic formula of EDB predicate containing all free variables of $φ$. 
	
	We now transform $\varphi_{r}(\vec{X_r})$ into a GNFO formula by structural induction on $\varphi_{r}(\vec{X_r})$. 
	Since GNFO is close under disjunction ($φ_1 ∨φ_2$), we transform each conjunction $\varphi_{r,i}(\vec{X_r})$ in the formula $\varphi_{r}(\vec{X_r})$ into a GNFO formula.
	We first group each negative predicate $\beta_{i,j}$ with its guard atom $w_{i,j}(\vec{Y})$. If a free variables $X$ appearing in $\beta_{i,j}$ but not in $w_{i,j}(\vec{Y})$, $X$ must appear in an equality $X=c$, we then substitute $c$ for $X$ in $\beta_{i,j}$ and obtain $	\varphi_{w_{i,j}}(\vec{Y}) \wedge \beta_{i,j}$, where $\vec{Y}$ contains all the free variable of $\beta_{i,j}$. If two negative predicates share the same guard atom then the guard atom can be used twice. 
	 \[\varphi_{r,i}(\vec{X_r}) = \exists \vec{E}_i, \left(\bigwedge\limits_{k} \beta_{i,k}\right) \wedge \left(\bigwedge\limits_{j} (\varphi_{w_{i,j}}(\vec{Y}) \wedge \beta_{i,j})\right)\]
	Because each $\beta_{i,k}$ in $\left(\bigwedge_{k} \beta_{i,k}\right) $ is a positive predicate, we inductively transform each $\beta_{i,k}$ into a GNFO formula.
	Now consider each formula $\psi = (\varphi_{w_{i,j}}(\vec{Y}) \wedge \beta_{i,j})$. 
	\begin{itemize}
		\item If $w_{i,j}$ is an EDB predicate, $\varphi_{w_{i,j}}(\vec{Y}) = w_{i,j}(\vec{Y})$, thus $\psi$ is a GNFO formula.
		\item If $w_{i,j}$ is an IDB predicate, by the construction of $\varphi_{w_{i,j}}(\vec{Y})$, we have $\varphi_{w_{i,j}}(\vec{Y}) = \bigvee\limits_l \varphi_{w_{i,j}^l}(\vec{Y}) $. As mentioned before in each $\varphi_{w_{i,j}^l}(\vec{Y}) $ there is an IDB atom $u_l(\vec{Z})$ containing all variables of $\vec{Y}$. Therefore,
		\begin{align*}
			 \psi &= \left(\bigvee\limits_l \varphi_{w_{i,j}^l}(\vec{Y})\right) \wedge  \beta_{i,j} \\
			 & \equiv  \bigvee\limits_l \varphi_{w_{i,j}^l}(\vec{Y}) \wedge  \beta_{i,j} \\
			 & \equiv  \bigvee\limits_l \varphi_{w_{i,j}^l}(\vec{Y}) \wedge  (\varphi_{u_l}(\vec{Z}) \wedge \beta_{i,j})
		\end{align*}
		We continue to inductively transform each $\varphi_{w_{i,j}^l}(\vec{Y})$ and $\varphi_{u_l}(\vec{Z}) \wedge \beta_{i,j}$ into a GNFO formula.
	\end{itemize}
	In this way, each formula $\varphi_{w_{i,j}}(\vec{Y}) \wedge \beta_{i,j}$ is transformed into a GNFO formula. Since GNFO is close under conjunction and existential quantifier, $\varphi_{r}(\vec{X_r})$ is transformed into a GNFO formula.
	
	We have constructed an equivalent GNFO formula $\varphi_{r}(\vec{X_r})$ for the Datalog query $(P,R)$. It is remarkable that in this transformation, we have introduced many predicate symbols $C_{<c}(X)$ and $C_{>c}(X)$ for comparison built-in predicates $<$ and $>$ in $P$. The introduction of new predicates $C_{<c}(X)$ and $C_{>c}(X)$ does not preserve the meaning of comparison symbols $<$ and $>$. Therefore, to reduce the satisfiability of Datalog query $(P,R)$ to the satisfiability of $\varphi_{r}(\vec{X_r})$, we need an axiomatization for the comparison built-in predicates. We construct a GNFO sentence for this axiomatization by using the similar technique for GN-SQL(\textsc{lin}) by B\'{a}r\'{a}ny et al. \cite{Barany2012}. Let the set of constant symbols in $P$ be $\{c_1, \ldots, c_n\}$, which is a finite subset of a totally ordered domain $dom$, with $c_1 < c_2 < \ldots < c_n$. The GNFO sentence that axiomatizes comparison built-in predicates is as follows:
	\[\Phi = ∀X, φ_{X<c_1} ∨ φ_{X=c_1} ∨ φ_{c_1<X<c_2} ∨ φ_{X=c_2} ∨ \ldots ∨ φ_{X>c_n}\]
	where
	\begin{align*}
		&φ_{X<c_1} =  \left\{
		\begin{array}{ll}
		\bigwedge\limits_{i \leq n} (C_{<c_i}(X) \wedge \neg (X=c_i) \wedge \neg C_{>c_i}(X)) \\
		\qquad \qquad \text{if $\exists c \in dom , c < c_1$}\\
		\bot \qquad \text{otherwise}
		\end{array}
		\right.\\
		&\begin{array}{ll}
			φ_{X=c_i} =  &(X=c_i ) \wedge \neg C_{<c_i}(X) \wedge  \neg C_{>c_i}(X) \wedge \\
				 &\left(\bigwedge_{j<i} ( C_{>c_j}(X) \wedge \neg (X = c_j )\wedge \neg C_{<c_j}(X) )\right) \wedge \\
					&\left(\bigwedge_{j>i} ( \neg C_{>c_j}(X) \wedge \neg (X = c_j) \wedge C_{<c_j}(X) )\right)
		\end{array}\\
		&φ_{c_i < X<c_{i+1}} =  \\
		&\qquad \quad \left\{
		\begin{array}{ll}
		\left(\bigwedge_{j\leq i} ( C_{>c_j}(X) \wedge \neg (X = c_j )\wedge \neg C_{<c_j}(X) )\right) \wedge \\
		\left(\bigwedge_{j>i} ( \neg C_{>c_i}(X) \wedge \neg (X = c_i) \wedge C_{<c_i}(X) )\right)\\
		\qquad \qquad \text{if $\exists c \in dom , c_i < c < c_{i+1}$}\\
		\bot \qquad \quad  \text{otherwise}
		\end{array}
		\right.\\
		&φ_{X>c_1} =  \left\{
		\begin{array}{ll}
		\bigwedge\limits_{i \leq n} ( \neg C_{<c_i}(X) \wedge \neg (X=c_i) \wedge C_{>c_i}(X)) \\
		\qquad \qquad \text{if $\exists c \in dom , c > c_n$}\\
		\bot \qquad \quad \text{otherwise}
		\end{array}
		\right.
	\end{align*}

	By this way, the Datalog query $(P,R)$ is satisfiable if and only if the GNFO sentence $\Phi \wedge \varphi_{r}(\vec{X_r})$ is satisfiable. Indeed, if there is a database $D$ over that the query $(P,R)$ is not empty, we can construct a signature $D'$ by copying all relations from $D$ and use all (finite) the suitable values of the active domain of $D$ to construct a relation corresponding to each predicate $C_{<c_i}(X)$/$C_{>c_i}(X)$. Clearly, $D'$ satisfies $\Phi$ and $\varphi_{r}(\vec{X_r})$. Conversely, if there is a signature $D'$ that satisfies $\Phi$ and $\varphi_{r}(\vec{X_r})$ we can construct a database $D$ by an isomorphic copy of all relations from $D'$ except the relations corresponding to predicates $C_{<c_i}(X)$ and $C_{>c_i}(X)$. It is known that for GNFO formulas, satisfiability over finite structure coincides with satisfiability over unrestricted structures. In other words, any structures satisfying the GNFO formula are finite. Therefore $D'$ is a finite structure, i.e. a database. Since the satisfiability of a GNFO sentence is decidable, the satisfiability of the Datalog query $(P,R)$ is also decidable.
\end{proof}

\subsection{Proof of Theorem~\ref{thm:gn-constraints}}
\begin{proof}
	As in Lemma~\ref{lemma:decidability}, we first transform a query $Q$ in nonrecursive GN-Datalog with equalities, constants and comparisons into an equivalent guarded negation first-order formula $\varphi_r(\vec{Y})$. The result of $Q$ over a database $D$ is not empty iff $D$ satisfies the sentence $\exists \vec{Y}, \varphi_r(\vec{Y})$.
	Let $\Sigma$ be a set of guarded negation constraints and $\sigma_i = ∀\vec{X_i}, Φ_i(\vec{X_i}) → \bot$ ($i \in [1,m]$) be a constraint in $\Sigma$, where $Φ_i(\vec{X_i})$ is a conjunction of (negative) atoms. Clearly, each $Φ_i(\vec{X_i})$ is a guarded negation formula since there is a \textit{guard} atom in the rule body $Φ_i(\vec{X_i})$. We rewrite $\sigma_i$ as an equivalent sentence $\sigma_i \equiv \neg \exists \vec{X_i}, Φ_i(\vec{X_i})$. Now, the query $Q$ is satisfiable under $\Sigma$ iff there exists a database $D$ satisfying all $\sigma_i$ such that $D$ satisfies $\exists \vec{Y}, \varphi_r(\vec{Y})$. This means that we need to check whether there exists a database $D$ such that $D$ satisfies all $\sigma_i$ and $\exists \vec{Y}, \varphi_r(\vec{Y})$:
	$D \models (\bigwedge^m_{i=1} \neg \exists \vec{X_i}, Φ_1(\vec{X_i})) \wedge (\exists \vec{Y}, \varphi_r(\vec{Y}))$.
	Note that there is no free variable in $\exists \vec{X_i}, Φ_i(\vec{X_i})$ ($i\in[1,m]$) and all $Φ_1, \ldots, Φ_m$ and $\varphi_r(\vec{Y})$ are GNFO formulas, the conjunction $(\bigwedge^m_{i=1} \neg \exists \vec{X_i}, Φ_1(\vec{X_I})) \wedge (\exists \vec{Y}, \varphi_r(\vec{Y}))$ is a GNFO formula. Thus, the problem now is reduced to the satisfiability of a GNFO formula, which is decidable.
\end{proof}

\subsection{Proof of Lemma~\ref{lem:view_derivation}}
\begin{proof}
	From Definition~\ref{def:put_validity}, we know that there exists a view definition $get^d$ that satisfies both \textsc{GetPut} and \textsc{PutGet} with the given valid $put$.
	Let $get$ be an arbitrary view definition satisfying \textsc{GetPut} with $put$, i.e., $put(S,get(S)) = S $ for any $S$.
	By applying the query $get^d$ to both the right-hand side and left-hand side of this equation, and using the \textsc{PutGet} property of $get^d$ and $put$, we obtain:
	\begin{align*}
	get^d(put(S,get(S))) = get^d(S) \Leftrightarrow get(S) = get^d(S)
	\end{align*}
	This means that $get(S) = get^d(S)$ for any $S$, i.e., $get$ and $get^d$ are the same. Thus, $get$ satisfies \textsc{PutGet}  with $put$.

	$~~$
\end{proof}

\subsection{Proof of Lemma \ref{lemma:solving_getput}}
Let $\langle r_1, \ldots, r_n \rangle$ be a source database schema and $S$ be a database instance of this schema, i.e., $S$ contains all relations $R_1, \ldots, R_n$ corresponding to the schema $r_1, \ldots, r_n$. Let $v$ be a view over the source database.
Let $\Sigma$ be a set of $m$ guarded negation constraints over the view and the source database; each constraint is of the form $\sigma_i = ∀\vec{X_i}, Φ_{\sigma_i}(\vec{X_i}) → \bot$.

Let us consider a LVGN-Datalog putback program $putdelta$ for the view $v$. $putdelta$ takes a (updated) view instance $V$ and the original source database $S$ to result in a delta $\Delta S$ of the source. $V$ is a steady state of the view if $\Delta S$ has no effect on the original $S$, i.e., $S \oplus \Delta S = S$. Recall that $\Delta S$ contains all the tuples that need to be inserted/deleted into/from each source relation $R_i$ ($i\in[1,n]$), represented by two sets $\deltaAdd{R_i}$ and $\deltaDel{R_i}$ for these insertions and deletions, respectively. $S \oplus \Delta S = S$ iff
\begin{equation}
	\deltaDel{R_i} \cap R_i = \deltaAdd{R_i} \setminus R_i = \emptyset, \forall i\in[1,n] \label{eqn:emptiness}
\end{equation}

Note that each $\deltaAdd{R_i}$/$\deltaDel{R_i}$ is the IDB relation corresponding to delta predicate $+r_i$/$-r_i$ in the result of the Datalog program $putdelta$ over the view and source database $(S,V)$. Since $putdelta$ is nonrecursive, we have an equivalent relational calculus query $\varphi_{-r_i}(\vec{X_i})$/$\varphi_{+r_i}(\vec{X_i})$ for each $\deltaAdd{R_i}$/$\deltaDel{R_i}$.
Equation (\ref{eqn:emptiness}) is equivalent to the condition that two relational calculus queries $\varphi_{-r_i}(\vec{X_i}) \wedge r_i (\vec{X_i}) $ and  $\varphi_{+r_i}(\vec{X_i}) \wedge \neg r_i (\vec{X_i})$ must be empty over the view and source database $(S,V)$. In other words, the first-order sentences $\exists \vec{X_i}, \varphi_{-r_i}(\vec{X_i}) \wedge r_i (\vec{X_i})$ and $\exists \vec{X_i}, \varphi_{+r_i}(\vec{X_i}) \wedge \neg r_i (\vec{X_i})$ are not satisfiable over the view and source database $(S,V)$. Combined with the constraint set $\Sigma$, a steady-state view $V$ satisfies $\Sigma$ and $S \oplus putdelta (S,V) = S$ iff:
\begin{equation}
    \left\{ \begin{array}{l}
        (S,V) \not\models  \exists \vec{X_i}, \varphi_{-r_i}(\vec{X_i}) \wedge r_i (\vec{X_i}) , i \in [1,n]\\ 
        (S,V) \not\models \exists \vec{X_i}, \varphi_{+r_i}(\vec{X_i}) \wedge \neg r_i (\vec{X_i}) , i \in [1,n] \\
        (S,V) \not \models  \exists \vec{X_i}, \Phi_{\sigma_i}(\vec{X_i}) , i \in [n+1,n+m]
    \end{array}
    \right.
\end{equation}
where $\vec{X_i}$ denotes a tuple of variables.
Note that $(S,V) \not\models \xi_1$ and $(S,V) \not\models \xi_2$ are equivalent to $(S,V) \not\models \xi_1 \vee \xi_2$. Thus, we have:
\begin{equation}
	\left\{ \begin{array}{l}
	(S,V) \not\models \exists \vec{X_i}, (\varphi_{-r_i}(\vec{X_i}) \wedge r_i (\vec{X_i})) \vee \\ 
	\qquad \qquad (\varphi_{+r_i}(\vec{X_i}) \wedge \neg r_i (\vec{X_i})) , i \in [1,n] \\
	(S,V) \not \models  \exists \vec{X_i}, \Phi_{\sigma_i}(\vec{X_i}) , i \in [n+1,n+m]
	\end{array}
	\right. 
	\label{eq:raw_getput}
\end{equation}
We now find such a $V$ satisfying ($\ref{eq:raw_getput}$).


\begin{claim}
    Given a putback program $putdelta$ written in LVGN-Datalog for a view $v$ and a source schema $\langle r_1, \ldots, r_n \rangle$, each relational calculus formula $\varphi_{r}(\vec{X_r})$ of the query $(putdel\-ta, R)$, where $R$ is an IDB relation corresponding to IDB predicate $r$ in $P$, can be rewritten in the following linear-view form:
    \[\left(\bigvee_{k=1}^{p} \exists \vec{E}_{1k}, v(\vec{Y}_{1k}) \wedge \psi_{1k}\right) \vee \left(\bigvee_{k=1}^{q} \exists  \vec{E}_{2k}, \neg v(\vec{Y}_{2k})  \wedge \psi_{2k}\right)\vee \psi_{3}\]
    where view atom $v$ does not appear in $\psi_{1k}$, $\psi_{2k}$ or $\psi_{3}$. 
    Each of the formulas $\exists\vec{E}_{1k}, v(\vec{Y}_{1k}) \wedge \psi_{1k}$, $\exists\vec{E}_{2k}, \neg v(\vec{Y}_{2k}) \wedge \psi_{1k}$ and $\psi_{3}$ is a safe-range GNFO formula and has the same free variables $\vec{X_r}$.
    \label{claim:view_predicate_norm}
\end{claim}
\begin{proof} 
The proof is conducted inductively on the transformation (presented in Subsection~\ref{sec:proof-decidability} - the proof of Lemma \ref{lemma:decidability}) between the Datalog query $(putdelta,R)$ and an equivalent GNFO formula $\varphi_{r}(\vec{X_r})$.
Note that in this transformation, each $\varphi_{r,i}$ is a safe-range\footnote{A first-order formula $\psi$ is a safe-range formula if all variables in $\psi$ are range restricted \cite{alicebook}. In fact, for each nonrecursive Datalog query with negation, there is an equivalent safe-range first-order formula, and vice versa \cite{alicebook}.} formula, i.e. is a relational calculus \cite{alicebook}.

We inductively prove that every $\varphi_r$ can be transformed into the linear-view form.
The base case is trivial.
For the inductive case,  due to the linear-view restriction, if $r$ is a normal predicate (not a delta predicate), then there is no view atom $v$ in all the rules defining $r$; thus, $\varphi_{r} = \bigvee_{i = 0}^m \varphi_{r,i}$ is in the linear-view form, where $\psi_{3} = \bigvee_{i = 0}^m \varphi_{r,i}$ and $p=q=0$. On the other hand, if $r$ is a delta predicate, in each $i$-th rule $r(\vec{X_r}) \ruleeq \alpha_{i,1}, \ldots, \alpha_{i,n_1}$, there are two cases.
The first case is that there is no $\alpha_{i,j_0}$ of a view atom $v$, $\varphi_{r,i} = \exists \vec{E}_i, \bigwedge_{j=0}^{n_i} \beta_{{i,j}}$ is in the linear-view form, where $\psi_{3} = \varphi_{r,i}$ and $p=q=0$.
In the second case, there is only one $\alpha_{i,j_0}$, which is an atom $v(\vec{Y_i})$ or a negated atom $\neg v(\vec{Y_i})$. Thus, $\varphi_{r,i} = \exists \vec{E}_i, v(\vec{Y_i}) \wedge \bigwedge_{j=0, j \neq j_0}^{n_i} \beta_{{i,j}}$ or $\varphi_{r,i} = \exists \vec{E}_i, \neg v(\vec{Y_i}) \wedge \bigwedge_{j=0, j \neq j_0}^{n_i} \beta_{{i,j}}$. Therefore, $\varphi_{r,i}$ is rewritten in the linear-view form. Note that if two formulas are in the linear-view form, then the disjunction of them can be transformed into the linear-view form. Indeed,
\begin{align*}
	& \left(\bigvee_{k=1}^{p_1} \exists \vec{E}_{1k}, v(\vec{Y}_{1k}) \wedge \psi_{1k}\right) \vee \left(\bigvee_{k=1}^{q_1} \exists \vec{E}_{2k}, \neg v(\vec{Y}_{2k})  \wedge \psi_{2k}\right) \\ 
	& \vee  \psi_{3} ~ \vee\\
    &  \left(\bigvee_{k=p_1+1}^{p_2} \exists \vec{E}_{1k}, v(\vec{Y}_{1k}) \wedge \psi_{1k}\right) \vee \\
    &  \qquad \qquad \qquad \qquad \left(\bigvee_{k=q_1+1}^{q_2} \exists \vec{E}_{2k}, \neg v(\vec{Y}_{2k})  \wedge \psi_{2k}\right)  \vee\psi'_{3}  \\
    &  \equiv \left(\bigvee_{k=1}^{p_2} \exists \vec{E}_{1k}, v(\vec{Y}_{1k}) \wedge \psi_{1k}\right) \vee  \\
    &  \qquad \qquad \qquad \qquad \left(\bigvee_{k=1}^{q_2} \exists \vec{E}_{2k}, \neg v(\vec{Y}_{2k})  \wedge \psi_{2k}\right) \vee (\psi_{3} \vee \psi'_{3}) 
\end{align*}
In this way, $\varphi_{r} = \bigvee\limits_{i = 1}^m \varphi_{r,i}$ is rewritten in the linear-view form.

As proven in \cite{Barany2012}, we can continue to transform each safe-range formula $\varphi_{r,i}$ into a GNFO formula. In other words, in the linear-view form of $\varphi_{r}$, each $\exists \vec{E}_{1k}, v(\vec{Y}_{1k}) \wedge \psi_{1k}$ and $\exists  \vec{E}_{2k}, \neg v(\vec{Y}_{2k})  \wedge \psi_{2k}$, and $\psi_{3}$ can be transformed into a safe-rage GNFO formula. In this transformation, if $\exists \vec{E}_{1k}, v(\vec{Y}_{1k}) \wedge \psi_{1k}$ is transformed into $\exists \vec{E}_{1k}, v(\vec{Y}_{1k}) \wedge \psi'_{1k} \vee v(\vec{Y}_{1k}) \wedge \psi''_{1k}$, we will transform it into $(\exists \vec{E}_{1k}, v(\vec{Y}_{1k}) \wedge \psi'_{1k}) \vee (\vec{E}_{1k}, v(\vec{Y}_{1k}) \wedge \psi''_{1k})$. In this way, we finally obtain a safe-range GNFO formula of $\varphi_{r}$, which is also in the linear-view form.
\end{proof}

%
We now know that the relational calculus formula $\varphi_{\pm r}$ of each delta predicate $\pm r$ is rewritten in the linear-view form.
For each constraint $\sigma_i$ of the form $∀\vec{X_i}, Φ_{\sigma_i}(\vec{X_i}) → \bot$, we can also transform the conjunction $ \Phi_{\sigma_i}(\vec{X_i})$ into the linear-view form. Indeed, let us consider a new Datalog rule in $putdelta$ as the following:
\[b_i(\vec{X_i}) \ruleeq \Phi_{\sigma_i}(\vec{X_i}).\]
in which the view is linearly used.
The conjunction $ \Phi_{\sigma_i}(\vec{X_i})$ is equivalent to the relational calculus query $\varphi_{b_i}(\vec{X_i})$ of relation $b_i$, which can be transformed into the linear-view form.

Since $\varphi_{\pm r}(\vec{X_r})$ can be rewritten in the linear-view form, the conjunction $\varphi_{\pm r}(\vec{X_r}) \wedge r(\vec{X_r})$ can be rewritten in the linear-view form by applying the distribution of existential quantifier over disjunction:
\begin{align*}
    & \varphi_{\pm r}(\vec{X_r}) \wedge r (\vec{X_r}) \equiv  \left(\bigvee_{k=1}^{p} r (\vec{X_r}) \wedge \exists \vec{E}_{1k}, v(\vec{Y}_{1k}) \wedge \psi_{1k}\right) \vee\\ 
    & \qquad  \left(\bigvee_{k=1}^{q} r (\vec{X_r}) \wedge \exists \vec{E}_{2k}, \neg v(\vec{Y}_{2k})  \wedge \psi_{2k}\right)\vee  (r (\vec{X_r}) \wedge \psi_{3})
\end{align*}
$\vec{X_r}$ is the free variable of $\varphi_{r}(\vec{X_r})$; hence, no existential variable in $\vec{E}_{1k}$ or $\vec{E}_{2k}$ is in $\vec{X_r}$. We can push $r (\vec{X_r})$ into the existential quantifier $\exists \vec{E}_{1k}$/$\exists \vec{E}_{2k}$ and obtain:
\begin{align*}
    & \left(\bigvee_{k=1}^{p}  \exists \vec{E}_{1k}, v(\vec{Y}_{1k}) \wedge r (\vec{X_r}) \wedge \psi_{1k}\right) \vee\\ 
    &  \left(\bigvee_{k=1}^{q}\exists \vec{E}_{2k}, \neg v(\vec{Y}_{2k}) \wedge r (\vec{X_r}) \wedge  \psi_{2k}\right)\vee  (r (\vec{X_r}) \wedge \psi_{3}) \\
\end{align*}
This is in the linear-view form. Therefore, the disjunction $(\varphi_{+ r}(\vec{X_r}) \wedge r(\vec{X_r})) \vee \varphi_{- r}(\vec{X_r}) \wedge r(\vec{X_r})$ can be rewritten in the linear-view form.
The constraint~(\ref{eq:raw_getput}) is now rewritten as:

\begin{align*}
    \left\{
    \begin{array}{rl}
        (S,V) \not\models&  \exists \vec{X_i}, \left(\bigvee\limits_{k=1}^{p_i} \exists \vec{E}^i_{1k}, v(\vec{Y}^i_{1k}) \wedge \psi^i_{1k}\right) \vee \\
        &\left(\bigvee\limits_{k=1}^{q_i} \exists  \vec{E}^i_{2k}, \neg v(\vec{Y}^i_{2k})  \wedge \psi^i_{2k}\right)\vee \psi^i_{3}, i\in [1,n]\\
        (S,V) \not\models&  \exists \vec{X_i}, \left(\bigvee\limits_{k=1}^{p_i} \exists \vec{E}^i_{1k}, v(\vec{Y}^i_{1k}) \wedge \psi^i_{1k}\right) \vee \\
		& \left(\bigvee\limits_{k=1}^{q_i} \exists  \vec{E}^i_{2k}, \neg v(\vec{Y}^i_{2k})  \wedge \psi^i_{2k}\right)\vee \psi^i_{3}, \\ 
		& \qquad \qquad \qquad \qquad\qquad\quad i\in [n+1,n+m]
    \end{array}
    \right.
\end{align*}

By applying the distribution of existential quantifier over disjunction
\[\exists \vec{X_i}, \xi_1(\vec{X_i}) \vee \xi_2(\vec{X_i}) \equiv (\exists \vec{X_i}, \xi_1(\vec{X_i})) \vee (\exists \vec{X_i}, \xi_2(\vec{X_i}))\] 
we have: 
\begin{align*}
    \left\{
    \begin{array}{rl}
        (S,V) \not\models&   \left(\bigvee\limits_{k=1}^{p_i} \exists \vec{X_i},\exists \vec{E}^i_{1k}, v(\vec{Y}^i_{1k}) \wedge \psi^i_{1k} \right) \vee \\
        & \left( \bigvee\limits_{k=1}^{q_i} \exists \vec{X_i},\exists   \vec{E}^i_{2k}, \neg v(\vec{Y}^i_{2k})  \wedge \psi^i_{2k} \right) \vee\\
        & \exists \vec{X_i},\psi^i_{3}, i\in [1,n]\\
        (S,V) \not\models& \left( \bigvee\limits_{k=1}^{p_i} \exists \vec{X_i}, \exists \vec{E}^i_{1k}, v(\vec{Y}^i_{1k}) \wedge \psi^i_{1k} \right) \vee \\
        & \left( \bigvee\limits_{k=1}^{q_i} \exists \vec{X_i}, \exists  \vec{E}^i_{2k}, \neg v(\vec{Y}^i_{2k})  \wedge \psi^i_{2k} \right) \vee\\
        & \exists \vec{X_i}, \psi^i_{3}, i\in [n+1,n+m]
    \end{array}
    \right.
\end{align*}
Here, we have disjunction of many formulas on the right-hand side, and we can apply the equivalence between $(S,V) \not\models \xi_1 \vee \xi_2$ and $((S,V) \not\models \xi_1) \wedge ((S,V) \not\models \xi_2) $ to separate the disjunction on the right-hand side and obtain $n_3$ sentences as follows:
\begin{align}
    \left\{
    \begin{array}{rl}
        (S,V) \not\models&  \exists \vec{E_k}, v(\vec{Y_k}) \wedge \psi_{k} , k\in [1,n_1] \\
        (S,V) \not\models& \exists \vec{E_k}, \neg v(\vec{Y_k})  \wedge \psi_{k} , k\in [n_1+1,n_2]\\
        (S,V) \not\models& \exists \vec{E_k},\psi_{k}, k\in [n_2+1,n_3]
    \end{array}
    \right.
    \label{eq:raw_view_derivation}
\end{align}
where $n_1 = \sum\limits _{i=1}^{n+m} p_i$, $n_2 =n_1+ \sum\limits_{i=1}^{n+m} q_i$ and $n_3 = n_2+ n+m$. All variables in $\vec{Y_k}$ are in $\vec{E_k}$ for any $k$.

Note that $\exists W, v(\vec{Y_k}) \wedge \psi_{k} \equiv v(\vec{Y_k}) \wedge \exists W, \psi_{k} $ if $W$ is not a free variable in $v(\vec{Y_k})$. In this way, we push existential variables in $\vec{E_k}$ but not in $\vec{Y_k}$, denoted by $\vec{Z_k}$, into the subformula $\psi_{k}$. In the case that there is a variable $X$ appearing more than once in $\vec{Y_k}$, we can introduce a new fresh variable $X'$ and add the equality $X=X'$ to the formulas after the quantifier $\exists \vec{Y_k}$. For example,
\[\exists Y_1 Y_1 Y_2, v(Y_1, Y_1, Y_2) \equiv \exists Y_1 Y_1' Y_2, v(Y_1, Y_1', Y_2) \wedge Y_1=Y_1'\]
We then substitute the variables in each $\vec{Y_k}$ to obtain the same $\vec{Y} = Y_1,\ldots, Y_{arity(v)}$ for each $\vec{Y_k}$. Then, we have $n_3$ FO sentences that $(S,V)$ must not satisfy:
\[
    \left\{ \begin{array}{l} 
        (S,V) \not\models  \exists \vec{Y}, v(\vec{Y}) \wedge \exists \vec{Z_k}, \psi_{k}(\vec{E_k}) , k \in [1,n_1]\\ 
        (S,V) \not\models \exists \vec{Y}, \neg v(\vec{Y})  \wedge \exists \vec{Z_k}, \psi_{k}(\vec{E_k}) , k \in [n_1+1,n_2] \\
        (S,V) \not\models \exists \vec{E_k}, \psi_{k}(\vec{E_k}), k \in [n_2+1,n_3]
    \end{array}
    \right.
\]
Because $((S,V) \not\models \xi_1) \wedge ((S,V) \not\models \xi_2)$ is equivalent to $(S,V) \not\models \xi_1 \vee \xi_2$, we have:
\[
    \left\{ \begin{array}{l} 
        (S,V) \not\models \bigvee\limits_{k=1}^{n_1} (\exists \vec{Y}, v(\vec{Y}) \wedge \exists \vec{Z_k}, \psi_{k}(\vec{E_k}))\\ 
        (S,V) \not\models \bigvee\limits_{k=n_1+1}^{n_2} (\exists \vec{Y}, \neg v(\vec{Y})  \wedge \exists \vec{Z_k}, \psi_{k}(\vec{E_k})) \\
        (S,V) \not\models \bigvee\limits_{k=n_2+1}^{n_3} (\exists \vec{E_k}, \psi_{k}(\vec{E_k}))
    \end{array}
    \right.
\] 
By applying the distribution of existential quantifier over disjunction $(\exists \vec{Y}, \xi_1(\vec{Y})) \vee (\exists \vec{Y}, \xi_2(\vec{Y})) \equiv \exists \vec{Y}, \xi_1(\vec{Y}) \vee \xi_2(\vec{Y}) $, we have: 

\[
    \left\{ \begin{array}{l} 
        (S,V) \not\models \exists \vec{Y},\bigvee\limits_{k=1}^{n_1} ( v(\vec{Y}) \wedge \exists \vec{Z_k}, \psi_{k}(\vec{E_k}))\\ 
        (S,V) \not\models \exists \vec{Y}, \bigvee\limits_{k=n_1+1}^{n_2} ( \neg v(\vec{Y})  \wedge \exists \vec{Z_k}, \psi_{k}(\vec{E_k})) \\
        (S,V) \not\models \bigvee\limits_{k=n_2+1}^{n_3} ( \exists \vec{E_k}, \psi_{k}(\vec{E_k}))
    \end{array}
    \right.
\] 
By applying the distribution of conjunction over disjunction $(p\wedge q) \vee (p\wedge r) \equiv p\wedge (q\vee r)$, we have:
\begin{align}
    &\left\{ \begin{array}{l}
        (S,V) \not\models  \exists \vec{Y}, v(\vec{Y}) \wedge \phi_1(\vec{Y})\\ 
        (S,V) \not\models \exists \vec{Y}, \neg v(\vec{Y})  \wedge \phi_2(\vec{Y}) \\
        (S,V) \not\models \phi_3
    \end{array}
    \right.\\
    \Leftrightarrow &\left\{ \begin{array}{l}
        (S,V) \models \forall \vec{Y}, v(\vec{Y}) \wedge \phi_1(\vec{Y}) \rightarrow \bot\\ 
    (S,V) \models \forall \vec{Y}, \neg v(\vec{Y})  \wedge \phi_2(\vec{Y}) \rightarrow \bot \\
        (S,V) \not\models \phi_3
    \end{array}
    \right.
    \label{eq:sentence_combination_apendix}
\end{align}
where $\phi_1 = \bigvee_{k=1}^{n_1} (\vec{Z_k}, \psi_{k}(\vec{E_k}))$, $\phi_2 =  \bigvee\limits_{k=n_1+1}^{n_2} ( \exists \vec{Z_k}, \psi_{k}(\vec{E_k}))$ and $\phi_3 = \bigvee\limits_{k=n_2+1}^{n_3} ( \exists \vec{E_k}, \psi_{k}(\vec{E_k}))$.

Note that in (\ref{eq:raw_view_derivation}), each $\exists \vec{E_k},\psi_{k}$ ($k\in [n_2+1,n_3]$) is a safe-range GNFO formula; hence, $\phi_3$ is a GNFO sentence. Each $\exists \vec{E_k}, \neg v(\vec{Y_k})  \wedge \psi_{k}$ ($k\in [n_1+1,n_2]$) is a safe-range GNFO formula, which means that each $ \psi_{k}$ ($k\in [n_1+1,n_2]$) is a safe-range GNFO formula; hence, $\phi_2$ is a safe-range GNFO formula. Each $\exists \vec{E_k}, v(\vec{Y_k}) \wedge \psi_{k} , k\in [1,n_1]$ is a safe-range GNFO formula; hence, $v(\vec{Y}) \wedge \phi_1(\vec{Y}) \equiv  \bigvee_{k=1}^{n_1} (\exists \vec{Y}, v(\vec{Y}) \wedge \exists \vec{Z_k}, \psi_{k}(\vec{E_k})) \equiv \bigvee_{k=1}^{n_1}(\exists \vec{E_k}, v(\vec{Y_k}) \wedge \psi_{k} ) $, which is a safe-range GNFO formula.


\subsection{Proof of Proposition~\ref{prop:putdelta_incrementalization}}
\begin{proof}
	Consider a database $S$ over schema $\langle r_1, \ldots, r_n\rangle$. 
	$S \oplus \Delta S = S$ means that
	$\deltaDel{R_i} \cap R_i = \emptyset$ and $\deltaAdd{R_i} \setminus R_i = \emptyset$ ($i\in[1,n]$).
	Let $\Delta^2 S$ be the change on $\Delta S$, i.e.,
    $\Delta^2 S$ contains insertions and deletions into/from each $\deltaAdd{R_i}$ and $\deltaDel{R_i}$.
    We use $\Delta^\pm_{R_i}$ as an abbreviation for $\deltaAdd{R_i}$ and $\deltaDel{R_i}$.
	Let $\Delta^+(\Delta^\pm_{R_i})$ and $\Delta^-(\Delta^\pm_{R_i})$ be the set of insertions and the set of deletions for $\Delta^\pm_{R_i}$, respectively.
	The new instance $\Delta'^\pm_{R_i}$ of each $\Delta^\pm_{R_i}$ in $\Delta S$ is obtained by:
	\[\Delta'^\pm_{R_i} = (\Delta^\pm_{ R_i} \setminus \Delta^-(\Delta^\pm_{R_i})) \cup \Delta^+(\Delta^\pm_{R_i})\]
	We finally obtain a new source database $S'$ by applying each $\Delta'^\pm_{R_i}$ in $\Delta S'$ to the corresponding relation $R_i$ in database $S$:
	\revision{\[
		\begin{array}{rl}
		R'_i = & (R_i \setminus \Delta'^-_{R_i}) \cup \Delta'^+_{R_i}\\
		= & (R_i \setminus ((\deltaDel{R_i} \setminus \Delta^-(\deltaDel{R_i})) \cup \Delta^+(\deltaDel{R_i}))) \\
		&\cup ((\deltaAdd{R_i} \setminus \Delta^-(\deltaAdd{R_i})) \cup \Delta^+(\deltaAdd{R_i}))
		\end{array}    
		\]}
	Because $\deltaDel{R_i}$ and $\deltaAdd{R_i}$ are disjoint, and because $\deltaDel{R_i} \cap R_i = \emptyset$ and $\deltaAdd{R_i} \setminus R_i = \emptyset$, we can simplify the above equation to:
	\revision{\begin{equation}
		\label{eq:putdelta_incrementalization}
		R'_i = R_i \setminus \Delta^+(\deltaDel{R_i}) \cup \Delta^+({\deltaAdd{R_i}})
		\end{equation}}
	Note that $\Delta^+(\deltaDel{R_i})$ and $\Delta^+({\deltaAdd{R_i}})$ contain all the tuples inserted into $\deltaDel{R_i}$ and $\deltaAdd{R_i}$, respectively. In other words, $\Delta^+(\deltaDel{R_i})$ and $\Delta^+({\deltaAdd{R_i}})$ are delta relations in $\Delta^{2+} S$.
	This means that the source database $S'$ is obtained by applying $\Delta^{2+} S$ to $S$:
	$S' =  S \oplus \Delta^{2+} S$.
\end{proof}

\subsection{Proof of Lemma~\ref{lem:LVGN-iincrementalization}}
\begin{proof}
Consider a valid LVGN-Datalog putback program $putdelta$ for a view $v$ and source database schema $\langle r_1, \ldots, r_n \rangle$.
Since $putdelta$ is in LVGN-Datalog, the view predicate occurs only in the rules defining delta relations of the source ($\pm r_1$, \ldots, $\pm r_n$), and at most once in each rule. When the view relation is changed, only delta relations, $\pm r_1, \ldots, \pm r_n$, are changed, all other relations (intermediate relations) in $putdelta$ are unchanged. Therefore, to incrementalize $putdelta$, we use only rules defining delta relations (having a predicate $\pm r_i$ as the head) to derive the rules computing changes to the delta relations.

A Datalog rule having a delta predicate $\pm r_i$ in the head and a view predicate $v$ in the body is in one of the following forms:
\begin{align*}
	\pm r_i (\vec{X}) &\ruleeq v(\vec{Y}), Q(\vec{Z}). \tag{positive view}\\
	\pm r_i (\vec{X}) &\ruleeq \neg~ v(\vec{Y}), Q(\vec{Z}). \tag{negative view}
\end{align*}
where $Q(\vec{Z})$ is the conjunction of the rest of the rule body. $Q(\vec{Z})$ is unchanged, whereas the view relation $v$ is changed to $v' = (v \setminus {-v} \cup {+v} )$, where $+v$ and $-v$ corresponds to the insertions set of deletions set, respectively. Similar to the incrementalization technique in \cite{Gupta1993}, by distributing joins over set minus and union we obtain 
\begin{align*}
+(\pm r_i) (\vec{X}) &\ruleeq +v(\vec{Y}), Q(\vec{Z}).\\
-(\pm r_i) (\vec{X}) &\ruleeq -v(\vec{Y}), Q(\vec{Z}). 
\end{align*}
for the case of positive view and 
\begin{align*}
+(\pm r_i) (\vec{X}) &\ruleeq -v(\vec{Y}), Q(\vec{Z}).\\
-(\pm r_i)(\vec{X}) &\ruleeq +v(\vec{Y}), Q(\vec{Z}). 
\end{align*}
for the case of negative view, where new delta relations is obtained by $\pm r_i' = (\pm r_i \setminus -(\pm r_i)) \cup +(\pm r_i)$.

Proposition~\ref{prop:putdelta_incrementalization} implies that the set of insertions to the delta relation, $+(\pm r_i)$, can be used as $\pm r_i'$ to apply to the source relation $r_i$ to obtain the same new source. Therefore, the rule computing $-(\pm r_i)$ is redundant, $\pm r_i'$ can be computed by the rules of $+(\pm r_i)$:
\begin{align*}
\pm r_i' (\vec{X}) &\ruleeq +v(\vec{Y}), Q(\vec{Z}).
\end{align*}
for the case of positive view and 
\begin{align*}
\pm r_i' (\vec{X}) &\ruleeq -v(\vec{Y}), Q(\vec{Z}).
\end{align*}
for the case of negative view. This shows that the transformation from origin $putdelta$ to an incremental one is substituting delta predicates of the view, $+v$ and $-v$, for positive and negative predicates of the view, $v$ and $\neg v$, respectively.
\end{proof}

\section{Transformation from safe-ra\-nge FO formula to Datalog}
\label{sec:fo-to-datalog}
In this section, we present the transformation from a safe-range FO formula $\varphi$ to an equivalent Datalog query. 

We first extend the algorithm that transforms a safe-range FO formula $\varphi$ into an equivalent formula in relational algebra normal form (RANF) described in \cite{alicebook} to allow built-in predicates ($<$ and $>$) occurring in $\varphi$.
Let us assume that $\varphi$ is in safe-range normal form (SRNF) in which there is no universal quantifiers, no implications, and there is no conjunction or disjunction sign that occurs directly below a negation sign. Every FO formula can be transformed into an SRNF formula by inductively applying the following logical equivalences: 
\begin{itemize}
	\item $∀\vec{x}ψ \equiv ¬∃\vec{x}¬ψ$
	\item $φ \to ψ \equiv \neg φ \vee ψ$
	\item $¬¬ψ \equiv ψ$
	\item $¬(ψ_1 ∨\ldots ∨ψ_n) \equiv (¬ψ_1 ∧\ldots ∧¬ψ_n)$
	\item $¬(ψ_1 ∧\ldots ∧ψ_n) \equiv (¬ψ_1 ∨ \ldots ∨¬ψ_n)$
\end{itemize}

The set of range-restricted variables of the FO formula $\varphi$ ($rr(\varphi)$) is inductively defined in the same way as \cite{alicebook}:
\begin{itemize}
	\item if $\varphi = R(e_1,\ldots{},e_n) $, $rr(\varphi)$ = the set of variables in $\{e_1,\ldots{},e_n\}$
	\item if $\varphi = (x = a)$ or $\varphi = (a = x)$, $rr(φ)={x}$
	\item if $\varphi = φ_1 ∧ φ_2$, $rr(φ) = rr(φ_1) ∪ rr(φ_2)$
	\item if $\varphi = \varphi_1 ∧ x=y $, $rr(φ)= rr(\varphi_1)$ if $\{x,y\} \cap rr(\varphi_1) = \emptyset$ and $rr(φ)= rr(\varphi_1) ∪ \{x, y\}$ otherwise
	\item if $\varphi = φ_1 ∨ φ_2$, $rr(ϕ) = rr(ϕ_1) ∩ rr(ϕ_2)$
	\item if $\varphi= \neg \varphi_1$, $rr(\varphi)=\emptyset$
	\item if $\varphi \in \{(x > a), (x < a), (x > y), (x < y) \}$, $rr(\varphi)=\emptyset$
	\item if $\varphi = \exists \vec{x} \varphi_1$, $rr(\varphi) = rr(\varphi_1) - \vec{x}$ if $\vec{x} \subseteq rr(\varphi_1)$ and $rr(\varphi) = \bot$ otherwise
\end{itemize}
where for each $Z$, $\bot \cup Z = \bot \cap Z = \bot - Z = Z - \bot = \bot$, $\bot$ indicates that some quantified variables are not range restricted.
Let $free(\varphi)$ be the set of free variables of $\varphi$.
$\varphi$ is a safe-range FO formula iff $rr(\varphi) = free(\varphi)$.

\begin{definition}[\cite{alicebook}]
	An occurrence of a subformula ψ in φ is self-contained if its root is ∧ or if
	\begin{itemize}
		\item $ψ=ψ_1∨ \ldots ∨ψ_n$ and $rr(ψ)=rr(ψ_1)=\ldots=rr(ψ_n)=free(ψ)$;
		\item $ψ = ∃\vec{x}ψ_1$ and $rr(ψ_1) = free(ψ_1)$;
	\end{itemize}
	A safe-range SRNF formula φ is in relational algebra normal form (RANF) if each subformula of φ 
	is self-contained.
\end{definition}

The algorithm that transforms a safe-range SRNF formula $\varphi$ into an equivalent RANF formula  is based on the following rewrite rules for each subformula $ψ$ in $\varphi$:
\begin{itemize}
	\item Push-into-or: If $ψ = ψ_1 ∧ \ldots ∧ ψ_n ∧ ξ$,
	where $ ξ = ξ_1 ∨ \ldots ∨ ξ_m$ and $rr(ψ) = free(ψ)$, but $rr(ξ) \neq free(ξ)$, we 
	nondeterministically choose a subset $\{i_1,\ldots,i_k\}$ of $\{1,\ldots,n\}$ such that
	\[ξ′ =(ξ_1 ∧ψ_{i_1} ∧\ldots ∧ψ_{i_k})∨\ldots ∨(ξ_m ∧ψ_{i_1} ∧\ldots ∧ψ_{i_k})\] satisfies 
	$rr(ξ′) = free(ξ′)$. Let $\{j_1, \ldots , j_l \} = \{1, \ldots, n\} \setminus \{i_1,\ldots,i_k\}$, we rewrite $\psi$ into $\psi'$:
	\[ψ′= ψ_{j_1} ∧\ldots ∧ψ_{j_l} ∧ξ′\]
	
	\item Push-into-quantifier: If
	$ψ = ψ_1 ∧ \ldots ∧ ψ_n ∧ ∃ \vec{x} ξ$ and $rr(ψ) = free(ψ)$, but $rr(ξ) \neq free(ξ)$, assuming that no variable in $\vec{x}$ is a free in $free(ψ_1 ∧ \ldots ∧ ψ_n)$ (this can be achieved by variable renaming), we 
	nondeterministically choose a subset $\{i_1,\ldots,i_k\}$ of $\{1,\ldots,n\}$ such that:
	\[ξ′=ψ_{i_1} ∧\ldots ∧ψ_{i_k} ∧ξ\]
	satisfies $rr(ξ′)=free(ξ′)$.
	We replace $ψ$ with \[ψ′=ψ_{j_1} ∧\ldots∧ψ_{j_l} ∧∃\vec{x}ξ′\]
	where $\{j_1,\ldots,j_l\}=\{1,\ldots,n\} \setminus \{i_1,\ldots,i_k\}$.
	
	\item Push-into-negated-quantifier: If
	$ψ = ψ_1 ∧\ldots ∧ ψ_n ∧ ¬ ∃ \vec{x} ξ $ and $rr(ψ) = free(ψ)$, but $rr(ξ) \neq free(ξ)$, assuming that no variable in $\vec{x}$ is a free in $free(ψ_1 ∧ \ldots ∧ ψ_n)$ (this can be achieved by variable renaming), we nondeterministically choose a subset $\{i_1,\ldots,i_k\}$ of $\{1,\ldots,n\}$ such that:
	\[ξ′=ψ_{i_1} ∧\ldots∧ψ_{i_k} ∧ξ\]
	satisfies $rr(ξ′)=free(ξ′)$.
	We replace $ψ$ with
	\[ψ′ =ψ_1 ∧\ldots∧ψ_n ∧¬∃\vec{x}ξ′\]
	$ψ′$ is equivalent to $ψ$ because the propositional formulas $p ∧ q ∧ ¬r$ and $p ∧ q ∧ ¬(p ∧ r)$ are equivalent. 
	And we continue to apply the Push-into-quantifier procedure
\end{itemize}

Now we transform the RANF formula $\varphi$ into an equivalent Datalog query $(P_\varphi,G_\varphi)$. Suppose $\{x_1, \ldots, x_k\} = free(\varphi)$, $(P_\varphi,G_\varphi)$ is inductively constructed as follows:
\begin{itemize}
	\item If $\varphi = R(e_1,\ldots,e_n)$, where $\{x_1, \ldots, x_k\}$ is the set of free variables in $\{e_1,\ldots,e_n\}$: 
	\[P_\varphi = \{G_\varphi(x_1, \ldots, x_k) \ruleeq R(e_1,\ldots,e_n).\}\]
	and the datalog query is $(P_\varphi,G_\varphi)$.
	\item If $\varphi$ is $x=a$ or $a=x$:  
	\[P_\varphi = \{G_\varphi(x) \ruleeq x = a.\}\]
	\item If $\varphi = \psi_1 \wedge \ldots \wedge \psi_m$,
	we divide  $\{\psi_1, \ldots, \psi_m\}$ into a set of positive subformulas $\{\psi_1, \ldots, \psi_{m_1}\}$ and a set of equalities/inequalities ($x=a$, $a=x$, $x=x$, $x=y$, $x>a$, $x<a$, $x>y$, $x<y$) $\{\psi_{m_1+1}, \ldots, \psi_{m_2}\}$, and a set of negative subformulas $\{\neg \psi_{m_2+1}, \ldots, \neg \psi_{m}\}$. Let $\{x_{i1},\ldots, x_{ik_i}\} = free(\psi_i)$,
	we inductively construct $(P_{\psi_i}, G_{\psi_i})$ for each $\psi_i$ in $\{\psi_1, \ldots, \psi_{m_1}\}$, and for each $\psi_i$ in $\{\psi_{m_2+1}, \ldots, \psi_{m}\}$.
	The Datalog query $(P_\varphi,G_\varphi)$ is as follows:
	\begin{align*}
	P_\varphi &= \left( \bigcup\limits^{m_1}_{i=1} P_{\psi_i} \right) \cup \left( \bigcup\limits^{m}_{i=m_2+1} P_{\psi_i}\right) \cup\\
	& \left\{ 
	\begin{array}{rl}
	G_\varphi(x_1, \ldots, x_k) \ruleeq & G_{\psi_1}(x_{11}, \ldots, x_{1k_1}), \ldots, \\
	&G_{\psi_{m_1}}(x_{m_11}, \ldots, x_{m_1k_{m_1}}), \\
	&\psi_{m_1+1}, \ldots, \psi_{m_2}, \\
	& \neg G_{\psi_{m_2+1}}(x_{(m_2+1)1}, \ldots, \\
	&x_{(m_2+1)k_{(m_2+1)}}), \ldots, \\
	&\neg G_{\psi_{m}}(x_{m1},\ldots, x_{mk_m}).
	\end{array}
	\right\}
	\end{align*}
	\item If $\varphi = ψ_1 ∨\ldots ∨ ψ_n$, where $ free(ψ_1) = \ldots = free(ψ_n) = \{x_1, \ldots, x_k\}$.
	We construct $(P_{\psi_i}, G)$ (with the same goal predicate $G$) for each $\psi_i$ in $\{ ψ_1 ,\ldots , ψ_n\}$
	and obtain:
	\[P_\varphi = \bigcup \limits^n_{i=1} P_{\psi_i}\]
	\[ G_\varphi = G \]
	\item If $\varphi = ∃y_1,\ldots,y_m, ψ(z_1,\ldots,z_n)$, where $\{x_1,\ldots,x_k\} = \{z_1,\ldots,z_n\} \setminus \{y_1,\ldots,y_m\}$: 
	\[P_\varphi = P_ψ \cup \{ G_\varphi(x_1,\ldots,x_k) \ruleeq G_\psi (z_1,\ldots,z_n).\}\]
\end{itemize}

To conclude that the transformation from safe-range FO formula to Datalog query is correct, i.e. $\varphi$ and $(P_\varphi,G_\varphi)$ are equivalent, we need to show that for any database instance $D$, $P_\varphi(D) | G_\varphi  = \{\vec{t} ~|~ D\models\varphi(\vec{t}) \}$, where $P_\varphi(D) | G_\varphi$ denotes the restriction of the output of $P$ over $D$ to the relation $G_\varphi$.
Indeed, let $D$ be fixed, by induction, we can show that for each subformula $\psi$ of $\varphi$ and each tuple $\vec{t}$, 
\[D\models\psi(\vec{t}) \Leftrightarrow P_\psi(D) \ni G_\psi(\vec{t}) \]

\section{Rules for incrementalizing pu\-tback programs}
\label{sec:incrementalization_rewite_rules}
\begin{figure}[t]
    \centering
    \small
    \begin{tabular}[t]{|l|}
        \hline
        Join and Selection \\ \hhline{|=|}
        \hspace{-7pt}$\begin{array}{l}
            h(\vec{X}) \ruleeq r_1(\vec{Y}), r_2(\vec{Z}).\\
            vars(\vec{X}) = vars(\vec{Y}) \cup vars(\vec{Z})
        \end{array}$\hspace{-7pt}\\ \hline
        \multicolumn{1}{|c|}{$\Downarrow$}\\ \hline
        \hspace{-7pt}$\begin{array}{l}
            -h(\vec{X}) \ruleeq {-r_1}(\vec{Y}) , r_2(\vec{Z}). \\
            {-h}(\vec{X}) \ruleeq r_1(\vec{Y}) , {-r_2}(\vec{Z}).\\
            {+h}(\vec{X}) \ruleeq {+r_1}(\vec{Y}), r_2^\nu(\vec{Z}).\\
            {+h}(\vec{X}) \ruleeq r_1^{\nu}(\vec{Y}) ,  {+r_2}(\vec{Z}).\\
            h^{\nu}(\vec{X}) \ruleeq r_1^{\nu}(\vec{Y}), r_2^{\nu}(\vec{Z}).
        \end{array}$\hspace{-7pt}
        \\ \hline
    \end{tabular}
    \begin{tabular}[t]{|l|}
        \hline
        Negation \\ \hhline{|=|}
        \hspace{-7pt}$\begin{array}{l}
            h(\vec{X}) \ruleeq r_1(\vec{X}), \neg r_2(\vec{Y}).\\
            vars(\vec{X}) \supseteq vars(\vec{Y})
        \end{array}$\hspace{-7pt} \\ \hline
        \multicolumn{1}{|c|}{$\Downarrow$}\\ \hline
        \hspace{-7pt}$\begin{array}{l}
            {-h}(\vec{X}) \ruleeq {-r_1}(\vec{X}) , \neg r_2(\vec{Y}). \\
            {-h}(\vec{X}) \ruleeq r_1(\vec{X}) ,  {+r_2}(\vec{Y}). \\
            {+h}(\vec{X}) \ruleeq {+r_1}(\vec{X}), \neg r_2^\nu(\vec{Y}).\\
            {+h}(\vec{X})\ruleeq r_1^{\nu}(\vec{X})  , {-r_2}(\vec{Y}).\\
            h^{\nu}(\vec{X}) \ruleeq r_1^{\nu}(\vec{X}), \neg r_2^{\nu}(\vec{Y}).
        \end{array}$\hspace{-7pt}
        \\ \hline
    \end{tabular}
    \begin{tabular}[t]{|l|}
        \hline
        Projection \\ \hhline{|=|}
        $h(\vec{X}) \ruleeq r_1(\vec{X},\vec{Y}).$ \\ \hline
        \multicolumn{1}{|c|}{$\Downarrow$}\\ \hline
        \hspace{-7pt}$\begin{array}{l}
            {+h}(\vec{X}) \ruleeq {+r_1}(\vec{X},\vec{Y}), \neg h(\vec{X}).\\
            {-h}(\vec{X}) \ruleeq {-r_1}(\vec{X}, \vec{Y}), \neg r_1^{\nu}(\vec{X}, \_).\\
            h^{\nu}(\vec{X}) \ruleeq r_1^{\nu}(\vec{X},\vec{Y}).
        \end{array}$\hspace{-7pt}
        \\ \hline
    \end{tabular}
    \begin{tabular}[t]{|l|}
        \hline
        Union \\ \hhline{|=|}
        \hspace{-7pt}$\begin{array}{l}
            h(\vec{X}) \ruleeq r_1(\vec{X}).\\
            h(\vec{X}) \ruleeq r_2(\vec{X}).
        \end{array}$\hspace{-7pt} \\ \hline
        \multicolumn{1}{|c|}{$\Downarrow$}\\ \hline
        \hspace{-7pt}$\begin{array}{l}
            {-h}(\vec{X}) \ruleeq {-r_1}(\vec{X}) , \neg r_2^\nu(\vec{X}). \\
            {-h}(\vec{X}) \ruleeq {-r_2}(\vec{X}) , \neg r_1^{\nu}(\vec{X}).\\
            {+h}(\vec{X}) \ruleeq {+r_1}(\vec{X}).\\
            {+h}(\vec{X}) \ruleeq {+r_2}(\vec{X}).\\
            h^{\nu}(\vec{X}) \ruleeq r_1^{\nu}(\vec{X}).\\
            h^{\nu}(\vec{X}) \ruleeq r_2^{\nu}(\vec{X}).
            \end{array}$\hspace{-7pt}
        \\ \hline
    \end{tabular}
    \caption{Rules for incrementalizing Datalog putback program.
        {$\vec{X}$ denotes a tuple of variables, $vars(\vec{X})$ denotes the set of all variables in $\vec{X}$}.}
    \label{fig:rewrite_rules}
\end{figure}
Given a putback program $putdelta$ in nonrecursive Datalog with negation (NR-Datalog$^{\neg}$), we shall derive Datalog rules to compute changes to delta relations of the source database when the view relation is changed. The derived Datalog rules form an incrementalized program of $putdelta$.

Our idea is that we first transform $putdelta$ into an equivalent Datalog program, in which every IDB relation is defined from at most 2 other relations. We then inductively apply the incrementalization rules in Figure~\ref{fig:rewrite_rules} to derive Datalog rules for computing changes to each IDB relation.

\begin{lemma}
	For every NR-Datalog$^\neg$ program $P$ with a goal IDB relation $R$, there is a NR-Datalog$^\neg$ program $P'$ in which each IDB relation is defined from at most two other relations such that the queries $(P,R)$ and $(P',R)$ are equivalent.
	\label{lem:twopredicate}
\end{lemma}
\begin{proof}(Sketch)
	There exists such a transformation between these two Datalog programs because a NR-Datalog$^\neg$ query $(P,R)$ is equivalent to a relational algebra expression, in which each binary relational operator can be simulated by Datalog rules with two relations in the rule bodies. 
\end{proof}

Considering the set semantics of the Datalog language, we propose rewrite rules (shown in Figure~\ref{fig:rewrite_rules}) for calculating changes to a relation $h$ which is defined from two relations $r_1$ and $r_2$. In each case of the definition of $h$, we derive Datalog rules that compute separately the set of insertions ($\Delta^+_{h}$) and the set of deletions ($\Delta^-_{h}$) to $h$ when there are changes to relations $r_1$ and $r_2$.
Note that in these derived Datalog rules, if $\Delta^+_{r_1}$ and $\Delta^-_{r_1}$ are disjoint, then the obtained $\Delta^+_{h}$ and $\Delta^-_{h}$ are also disjoint. Therefore, we can inductively apply the four incrementalization rules when $h$ is used to define other IDB relations. We have formally proven the correctness of these incrementalization rules by using an assistant theorem prover, stated as the following.
\begin{lemma}
    For each case in Figure~\ref{fig:rewrite_rules}, the new relation $h^\nu$ computed from its defining rules is the same as the result obtained by applying delta relations $+h$ and $-h$ computed by the derived Datalog rules to the original relation $h$.
    \label{lem:delta_rule}
\end{lemma}

Our incrementalization rules can be easily extended for built-in predicates (e.g., $=,<,>$) in the Datalog program by considering these predicates as unchanged relations in our incrementalization rules.

\section{Deriving view deltas}
\label{sec:deriveing_view_deltas}
\begin{algorithm}[th]
	\caption{\textsc{View-Delta}($u_1, \ldots, u_n$)}
	\label{algo:delta_view}
	\DontPrintSemicolon
	$\deltaAdd{V} \leftarrow \emptyset$; $\deltaDel{V} \leftarrow \emptyset$;\\
	\For{\textbf{each} DML statement $u$ in $u_1, \ldots, u_n$}{
		Derive the set $\delta^+$/$\delta^-$ of inserted/deleted tuples;
		
		$\deltaAdd{V} \leftarrow (\deltaAdd{V} \setminus \delta^-) \cup \delta^+$;
		
		$\deltaDel{V} \leftarrow (\deltaDel{V} \setminus \delta^+) \cup \delta^- $;
	}
\end{algorithm}
Our incrementalization on putback transformation requires deriving a delta relation $\Delta V$ of the view $V$ in the form of insertions and deletions when there are any view update requests.
In RDBMSs, these update requests are declarative DML (data manipulation language) statements of the following forms \cite{Ramakrishnan1999}: \texttt{INSERT INTO $V$ VALUES($\ldots$)},  \texttt{DELETE FROM V WHERE <condition>}, and \texttt{UPDATE V SET attr=expr,  ... WHERE <condition>}.
Fortunately, it is trivial to obtain from the \texttt{INSERT}/\texttt{DELETE} statement the tuples that need to be inserted or deleted. Meanwhile, an \texttt{UPDATE} statement on the view can be represented as deletions followed by insertions; hence, we can also derive the deleted/inserted tuples.

A view update request can be a sequence of DML statements rather than a single one. This sequence is combined into one transaction by using the SQL command \texttt{BEGIN} before the sequence and the command \texttt{END} after the sequence.
To address this case, we propose a procedure for calculating $\deltaAdd{V}$ and $\deltaDel{V}$ of the whole view update transaction, as shown in Algorithm~\ref{algo:delta_view}. Concretely, for each DML statement in the sequence, we derive the insertion set $\delta^+$ and the deletion set $\delta^-$, and we merge these changes to $\deltaAdd{V}$ and $\deltaDel{V}$.
In this way, later statements have stronger effects than earlier statements.
For example, if the sequence is inserting a tuple $\vec{t}$ and then deleting this tuple, $\vec{t}$ is no longer inserted, i.e., we remove $\vec{t}$ from $\deltaAdd{V}$.
\end{appendix}

\balance

\end{document}